\documentclass[oupdraft]{bio}
\usepackage{url}

\RequirePackage[OT1]{fontenc}
\RequirePackage{amsthm,amsmath}
\RequirePackage{natbib}
\usepackage{xr}
\usepackage[english]{babel}
\usepackage{booktabs}
\usepackage{longtable}
\usepackage{array}
\usepackage{multirow}
\usepackage[table]{xcolor}
\usepackage{float}
\usepackage{tabu}
\usepackage{algorithm}
\usepackage{algorithmic}
\usepackage{enumitem}
\setitemize{wide}

\usepackage{cleveref}

\usepackage{graphicx}
\usepackage[colorinlistoftodos]{todonotes}
\usepackage{float}
\usepackage[caption = false]{subfig}
\graphicspath{{./figs/}}
\usepackage{amssymb}

\catcode`\^ = 13 \def^#1{\sp{#1}{}}

\definecolor{Pink}{rgb}{1.0, 0.5, 0.5}
\definecolor{Maroon}{rgb}{0.8, 0.0, 0.0}

\def\boxit#1{\vbox{\hrule\hbox{\vrule\kern6pt\vbox{\kern6pt#1\kern6pt}\kern6pt\vrule}\hrule}}

\newcommand{\M}{\mbox}

\newtheorem{corollary}[theorem]{Corollary}

\newcommand{\bM}{\mbox{\bf M}}
\newcommand{\bN}{\mbox{\bf N}}

\newcommand{\bdelta}{\mbox{\boldmath $\delta$}}
\newcommand{\bbeta}{\mbox{\boldmath $\beta$}}
\newcommand{\btheta}{\mbox{\boldmath $\theta$}}
\newcommand{\bepsilon}{\mbox{\boldmath $\epsilon$}}

\newcommand{\bgamma}{\mbox{\boldmath $\gamma$}}

\newcommand{\what}{\widehat}

\def\beqn{\begin{eqnarray}}
\def\eeqn{\end{eqnarray}}
\def\beqns{\begin{eqnarray*}}
\def\eeqns{\end{eqnarray*}}

\def\0{{\bf 0}}

\def\a{{\bf a}}
\def\b{{\bf b}}
\def\C{{\bf C}}
\def\c{{\bf c}}
\def\D{{\bf D}}
\def\d{{\bf d}}

\def\h{{\bf h}}

\def\H{{\bf H}}
\def\I{{\bf I}}

\def\r{{\bf r}}

\def\R{{\bf R}}
\def\T{{\bf T}}
\def\bO{{\bf O}}
\def\bP{{\bf P}}

\def\U{{\bf U}}

\def\S{{\bf S}}
\def\s{{\bf s}}
\def\u{{\bf u}}

\def\W{{\bf W}}
\def\w{{\bf w}}
\def\X{{\bf X}}

\def\T{{\bf T}}

\def\y{{\bf y}}
\def\Z{{\bf Z}}
\def\z{{\bf z}}
\def\1{{\bf 1}}

\def\trans{^{\rm T}}

\newcommand{\bbR}{\mathbb{R}}

\newcommand{\mcJ}{\mathcal{J}}
\newcommand{\mcA}{\mathcal{A}}

\DeclareMathOperator*{\argmin}{arg\,min}

\newcommand{\bs}{\boldsymbol}

\newcommand{\beginsupplement}{%
        \setcounter{table}{0}
        \renewcommand{\thetable}{S\arabic{table}}%
        \setcounter{figure}{0}
        \renewcommand{\thefigure}{S\arabic{figure}}%
     }

\AtBeginDocument{\renewcommand{\abstractname}{}}

\begin{document}

\title{Robust regression with compositional covariates}

\author{ADITYA K. MISHRA$^\ast$, CHRISTIAN L. M\"ULLER\\[4pt]
\textit{Center for Computational Mathematics, 
Flatiron Institute, 162 5th Avenue,  NY 10010}
\\[2pt]
{amishra@flatironinstitute.org}}

\markboth%
{A. K. Mishra and C. L. M\"uller}
{RobRegCC}
\maketitle

\footnotetext{To whom correspondence should be addressed.}

\begin{abstract} 
{Many biological high-throughput data sets, such as  targeted amplicon-based and metagenomic sequencing data, 
are compositional in nature. A common exploratory data analysis task is to infer statistical associations 
between the high-dimensional microbial compositions and habitat- or host-related covariates. 
We propose a general robust statistical regression framework, \textsf{RobRegCC} (Robust Regression with Compositional
Covariates), which extends the linear log-contrast model by a mean shift formulation 
for capturing outliers. \textsf{RobRegCC} includes sparsity-promoting convex and non-convex 
penalties for parsimonious model estimation, a data-driven
robust initialization procedure, and a novel robust cross-validation model selection scheme. 
We show \textsf{RobRegCC}'s ability to perform simultaneous sparse log-contrast regression and 
outlier detection over a wide range of simulation settings and provide theoretical non-asymptotic 
guarantees for the underlying estimators. To demonstrate the seamless applicability of 
the workflow on real data, we consider a gut microbiome data set from HIV patients and 
infer robust associations between a sparse set of microbial species and host immune response 
from soluble CD14 measurements. All experiments are fully reproducible and available on GitHub at \url{https://github.com/amishra-stats/robregcc}.
}
{Compositional data; Microbiome; Robust regression; Mean shift; Sparsity; non-convexity}
\end{abstract}

\section{Introduction}
Many scientific data measurements are compositional in nature. Prominent examples include chemical
composition measurements of rocks and sediments in geology and relative abundances of sequencing
reads in microbial ecology. For instances, targeted amplicon sequencing (TAS) and metagenomic
profiling provides genomic survey data of microbial communities in their natural habitat, ranging
from marine ecosystems to the human gut \citep{Huttenhower2012, Thompson2017, Sunagawa2016,
McDonald2018}. These microbiome surveys typically comprise sparse relative (or compositional) counts of
operational taxonomic units (OTUs) or amplicon sequence variants (ASVs) \citep{Callahan2017,Edgar2016} 
and are often accompanied by measurements of additional covariates that characterize the underlying
habitat or the phenotypic status of the host. 

An important step in exploratory microbiome data analysis is the inference of parsimonious and robust
statistical relationships between the microbial compositions and habitat- or host-specific
measurements. Standard linear regression modeling can, however, not be applied in
this context because the microbial count data only carry relative or compositional information. 
Several regression techniques have been introduced to handle compositional data, including Dirichlet
multinomial mixture modeling \citep{Holmes2012} and kernel penalized regression \citep{Randolph2015}. 
A popular approach to regression modeling with compositional covariates is  
log-contrast regression, put forward by \cite{aitchison1984log} in the context of 
experiments with mixtures. In the linear log-contrast model, the continuous response is expressed 
as linear combination of log-transformed compositions subject to a zero-sum constraint on the
regression vector. This model allows the intuitive interpretation of the response as a linear
combination of log-ratios of the original compositions. An alternative equivalent low-dimensional 
approach considers linear regression after applying an isometric log-ratio (ilr) transform to the
compositions \citep{Hron2012}.

For microbiome data analysis, the linear log-contrast model has been brought
to the high-dimensional setting via regularization, e.g, via $\ell_1$ penalization \citep{Lin2014}
or more general structured sparsity approaches \citep{shi2016regression,Wang2017,Sun2018}. 
A related approach is the selection of balance (selbal) approach \citep{rivera2018balances}
which performs sparse greedy covariate selection on ilr transformed variables. 
While these approaches can lead to parsimonious models linking (microbial) compositions 
to responses of interest, they are sensitive to outliers or high-leveraged data points in the response. 

In this contribution, we alleviate these shortcomings by introducing \textbf{Rob}ust log-contrast 
\textbf{Reg}ression estimators with \textbf{C}ompositional \textbf{C}ovariates 
(\textsf{RobRegCC}), a novel robust regression modeling framework for compositional data. Figure~\ref{fig:workflow} shows the general \textsf{RobRegCC} workflow. 

\begin{figure}[h]
\begin{center}
\includegraphics[page = 1, width=1\textwidth, angle=0]{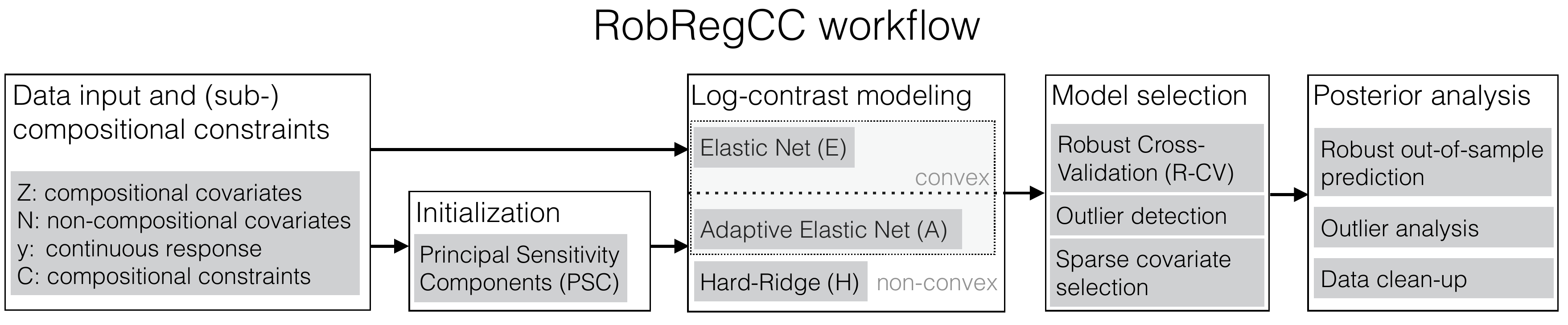}
\caption{The \textsf{RobRegCC} workflow for robust regression with compositional covariates.  }\label{fig:workflow}
\end{center}
\end{figure}

\textsf{RobRegCC} integrates a mean shift formulation in linear log-contrast regression which enables
the modeling of outliers in the response variable. The approach achieves 
parsimonious model identification, i.e., simultaneous outlier detection and variable selection,
through the integration of sparsity-promoting convex and non-convex regularizers. 
\textsf{RobRegCC} includes three different penalization approaches, the standard Elastic Net(E) penalty, a novel
adaptive Elastic Net(A) penalty, and a non-convex hard-ridge (H) penalty, resulting in a family of 
robust estimators. We derive theoretical guarantees for these estimators in the non-asymptotic 
setting. The latter two estimation procedures require initial parameter estimates which we provide
via principal sensitivity component (PSC) analysis \citep{pena1999fast}, adapted to the compositional
setting. We formulate the associated (non-)convex optimization problems using an augmented Lagrangian
framework and present an iterative thresholding/proximal algorithm for efficient numerical
minimization. \textsf{RobRegCC} also includes novel robust model selection and robust out-of-sample
prediction measures which may be of independent interest. For model selection, we put forward a robust 
cross-validation (R-CV) scheme which computes test sample error on ``clean" leave-out data using a
specifically tailored robust test statistics. The same statistic is also used to perform robust
out-of-sample prediction. All presented simulation and real-world experiments and computations
are available in a reproducible workflow on GitHub at \url{https://github.com/amishra-stats/robregcc}. \textsf{RobRegCC} is available on CRAN at \url{https://CRAN.R-project.org/package=robregcc}. 

\section{Robust log-contrast regression for compositional data}
High-throughput next-generation sequencing techniques typically provide  
read count data of the form $\D = [\d_1,\ldots,\d_n]\trans \in \bbR^{n \times p}$, comprising 
$n$ observations of a $p$-dimensional vector of read counts. 
The counts correspond, for instance, to the estimated number of OTUs, ASVs, or genes in a biological sample. 
Due to experimental limitations, the read counts only carry relative or proportional information 
and do not represent absolute abundances. One way to normalize these count data is to divide each sample 
by its total sum, resulting in a matrix $\W =[\w_1,\ldots,\w_n]\trans  \in \bbR^{n \times p}$ where each $\w_i =
\d_i/\1_p\trans\d_i$ represents a $p$-dimensional vector of proportions or compositions. 
Prior to the normalization, any zero count is replaced by a constant pseudo-count
\citep{aitchison1982statistical} or a small random count generated from an 
appropriate probability distribution \citep{sparcc}. Any compositional vector $\w_i$ is 
thus constraint to the $(p-1)$-dimensional simplex 
$\S^{(p-1)} =  \{[s_{1},\ldots,s_{p}]\trans: 0<s_{k}\le 1, \sum_{k=1}^p s_{k} = 1 \}$. 
The problem of interest is to find linear associations between the measured compositions $\W$ and
a continuous response or outcome variable of interest $\y = [y_1,\ldots,y_n]\trans \in \bbR^{n}$
that has been jointly collected with the relative abundance data. \cite{aitchison1984log} provide a 
useful framework to model such associations via log-contrast regression.

\subsection{The standard log-contrast regression model}
\label{subsec:logcontm}
The principle idea of log-contrast regression is to model the outcome $\y$ as linear 
combination of log-ratios derived from the compositional covariate data $\W$. A common
transform is the additive log-ratio (alr) transform \citep{aitchison1982statistical} which 
requires the choice of a reference. When considering the $k$th predictor as reference, 
the alr-transformed data are $\U= [\u_1,\ldots,\u_n]\trans$, where $\u_i = [u_{i1}, \ldots , 
u_{ip}]$ with $u_{ij} = \log(w_{ij}/w_{ik})$. 
The log-contrast regression model is written as
\begin{equation}
\label{eq:logconm}
y_i = \u_{i,- k} \,\, \b_{-k} +\epsilon_i , \qquad i = 1,\ldots,n 
\end{equation}
where $\b  = [b_{1} \, \ldots \, b_{p}]\trans$ is the coefficient vector, 
and $\bepsilon = [\epsilon_1 \, \ldots \,\epsilon_n] \in \bbR^{n}$ is independent and 
identically distributed (\textsf{IID}) noise with mean $\mathbb{E}(\epsilon_i) = 0$ and 
variance $\mathrm{Var}(\epsilon_i) = \sigma^2$. 
The symbol $-k$ denotes the exclusion of $kth$ entries in $\u_{i}$ and $\b$. 
A major drawback of model \eqref{eq:logconm} is its loss of permutation invariance
due to the choice of a reference \citep{aitchison1982statistical}. 
By expressing $b_k = - \sum_{i \neq k} b_i$, we can reformulate model \eqref{eq:logconm} 
into a symmetric permutation-invariant form as
\begin{align}
\label{eq:logcon_clr}
y_i = \z_i\trans \b+\epsilon_i , \qquad \1_p\trans \b = 0,  \qquad i = 1,\ldots,n \,,  
\end{align}
where $\z_i = [z_{i1} \, \ldots \, z_{ip}]\trans$ are log-transformed predictors with 
$z_{ij} = \log(w_{ij})$ \citep{aitchison1984log}. The linear constraint in 
\eqref{eq:logcon_clr} 
ensures that, after model fitting, the response can be equivalently expressed as linear 
combinations of log-ratios of the original compositions 
\citep{Aitchison2003,Sun2018,Bates2018,Combettes2020a}. The model also ensures 
subcompositional coherence, a key principle in compositional data analysis. This principle 
states that the analysis should be coherent even if we had only selected subcompositions 
out of the full compositions, or if the analyzed compositions are only parts of larger 
compositions containing other parts.

When additional $m$ non-compositional covariates $\bN \in \mathbb{R}^{n \times m}$, such as   
habitat and host-associated factors or other control variables are available, we can extend the 
linear log-contrast model to
\begin{align}
\label{eq:logcotLC}
\y =   \Z\b + \bN\a + \bepsilon , \qquad \1_p\trans \b = 0,
\end{align}
where $\Z = [\z_1,\ldots,\z_n]\trans$, $\b  \in \mathbb{R}^{p}$ is the 
coefficient vector for the compositional covariates, and $\a \in \mathbb{R}^{m}$ is the coefficient vector for all non-compositional variables, respectively. This 
model also allows to include an unconstrained intercept in the linear log-contrast model by
taking the first column of $\bN$ to be $\1_n$, the $n\times 1$ vector of ones.

The zero-sum constraint in \eqref{eq:logcotLC} can be generalized when grouping 
information about the predictors is available. In the microbiome context, each predictor 
can be associated with taxonomic or phylogenetic information, 
typically encoded in a taxonomic or phylogenetic tree $\mathcal{T}_{K,p}$ 
with $p$ leaves and $K$ levels. Following \citet{shi2016regression}, we can include this 
information in \eqref{eq:logcotLC} via a linear equality constraint. For instance,
when analyzing microbiome data at a fixed (taxonomic or phylogenetic) 
level of the tree, e.g., at the phylum level, this level induces a grouping of the $p$ taxa into 
$k$ disjoint sets with column index set ${\mathbb{A}_r}$ such that $|\mathbb{A}_r| = p_r$ 
for $r=1,\ldots,k$ and $\sum_{r=1}^k p_r = p$. Each set represents the taxa in the 
respective phylum. If the goal of the analysis is to be subcompositionally coherent with 
respect to the phylum groups, we can define the subcomposition matrix $\C_s$: 
\begin{align}
\C_s^\trans = \begin{bmatrix} \c_1 & \c_2 & \c_3 & \ldots & \c_k \end{bmatrix}\trans = \begin{bmatrix}
    \1_{p_1}\trans & \0 & \dots  & \0 \\
    \0  & \1_{p_2}\trans &  \dots  & \0 \\
    \vdots & \vdots &  \ddots & \vdots \\
    \0 & \0 &  \dots  & \1_{p_k}\trans 
\end{bmatrix}_{k \times p} \,,
\label{eq:defCmat}
\end{align}
where $ \c_{j}$ accounts for the composition in the subgroup with index set 
${\mathbb{A}_j}$ such that $ (\c_{j})_{\mathbb{A}_j} = \1_{p_j}$.

The model in \eqref{eq:logcotLC} can thus be generalized by including the subcomposition matrix $\C_s$:
\begin{align}
\label{eq:cp-slr}
\y =  \Z\b + \bN\a + \bepsilon = \sum_{r=1}^{k} 
\Z_{\mathbb{A}_r}\b_{\mathbb{A}_r} + \bN\a + \bepsilon, \qquad \M{s.t.} \qquad \C_s\trans \b = 
\0 \,, 
\end{align}
where $\{\Z_{\mathbb{A}_r} , \b_{\mathbb{A}_r}\}$ are the covariates and unknown 
coefficients corresponding to the $r$th sub-group. The model in \eqref{eq:logcotLC} is 
a special case of model \eqref{eq:cp-slr} with $\C_s = \1_p$. 

\subsection{Robust log-contrast regression model}
Many biological datasets, including microbiome profiling data, contain outliers or other 
forms of data corruptions that can hamper statistical estimation. For example, the extended 
log-contrast model in \eqref{eq:cp-slr} assumes errors $\bepsilon$ to be well behaved, i.e., 
free from outliers in $[\y,\Z,\bN]$. Following earlier work for wavelet estimation in partial
linear model \citep{Antoniadis2007,Gannaz2007} and linear regression
\citep{she2011outlier,Lee2012b,Nguyen}, we propose to extend 
the log-contrast model in \eqref{eq:cp-slr} with a mean shift vector 
$\bgamma = [\gamma_1 , \ldots  , \gamma_n]\trans$, accounting 
for the grossly corrupted observations in $y$, resulting in the model
\begin{align}
\label{eq:cp-Rslr}
\y =   \Z\b + \bN \a + \bgamma + \bepsilon , \quad \M{s.t.} \quad \C_s\trans\b= \0. 
\end{align}
The support set $\bs\mcJ(\bgamma)$ of the vector $\bgamma$ can thus capture potential
outliers in the response $y$. By fusing the compositional and non-compositional 
covariates into the general design matrix $\X = [\Z \,\, \bN ]$, we can denote the
corresponding model coefficients by $\bbeta = (\b \,\,\a) \in \mathbb{R}^p$. 
Augmenting the linear constraint matrix $\C_s$ by a $k \times m$ zero-matrix, denoted by
$\C = [\C_s \trans \,\, \0_{k \times m}]\trans$, the model in \eqref{eq:cp-Rslr} 
simplifies to
\begin{align}
\label{eq:cp-Rslr2}
\y =   \X \bbeta + \bgamma + \bepsilon , \quad \M{s.t.} \quad \C\trans\bbeta = \0 \,. 
\end{align}
This model forms the basis for the \textbf{Rob}ust log-contrast 
\textbf{Reg}ression estimators with \textbf{C}ompositional \textbf{C}ovariates 
(\textsf{RobRegCC}), considered in the remainder of the paper.

\subsection{Regularization for parameter estimation}
As the \textsf{RobRegCC} model in \eqref{eq:cp-Rslr2} is over-specified even in the
low-dimensional setting, comprising $(p+m+n-k)$ unknown parameters, we introduce a family of
regularized estimators using sparsity-inducing penalties. The proposed class of estimators for 
the parameters $(\bgamma, \bbeta)$ are associated with the following 
general optimization problem:     
\begin{align}
(\what{\bgamma},\what{\bbeta}) \equiv \argmin_{\bgamma,\bbeta} \Bigg\{ \frac{1}{2n} \| \y
- \X \bbeta - \bgamma \|_2^2 & + P_{\lambda_1}^1(\bgamma) + P_{\lambda_2}^2(\bbeta) \Bigg\}   
  \quad \M{s.t.} \quad \C\trans \bbeta = \0 \,, 
\label{eq:objfun21} 
\end{align}
where $P_{\lambda_1}^1(\bgamma)$ and  $P_{\lambda_2}^2(\bbeta)$ are sparsity-inducing 
regularizers with tuning parameters $\lambda_1$ and $\lambda_2$, respectively. The
regularization framework involves solving the optimization problem \eqref{eq:objfun21}  over a
grid of tuning parameters $\{\lambda_1,\lambda_2\}$. Our theoretical results (see Theorem
\ref{sec:non-asy:th1} in Section \ref{sec:nasyanal}) show that optimal tuning can be
achieved by setting $\lambda_1 = A\sigma(\log(en))^{1/2}$ and $\lambda_2 =
A\sigma(\log(ep))^{1/2}$ for some constant $A$. This motivates the introduction of a single
tuning parameter $\lambda$ and expressing $\lambda_1 =  k_1\lambda$ and $\lambda_2 = k_2
\lambda$ where $k_1 = \sqrt{\log(en)} $ and $k_2 = \sqrt{\log(ep)}$.

Normalizing the $\ell_2$-norm of the columns of $\X$ to  $\sqrt{n}$, scaling the mean shift
vector $\bgamma$ by the factor $\sqrt{n}$, and concatenating the unknowns into $\bs \delta
=[\delta_1, \ldots, \delta_{n+p}]\trans =  [\bgamma\trans \, \bbeta\trans]\trans$ leads to
a compact reformulation of \eqref{eq:objfun21} of the form: 
\begin{align}
\label{eq:objfun}
(\what{\bgamma},\what{\bbeta}) \equiv \argmin_{\bgamma,\bbeta} \Bigg\{\frac{1}{2n} \| \y  
- \X \bbeta - \sqrt{n}\bgamma \|_2^2  +  P_{\lambda}(\bs\delta) \Bigg\} \quad \M{s.t.} \quad \C\trans \bbeta = \0,
\end{align}
with $P_\lambda(\bs \delta) =  P_{\lambda_1}^1(\bgamma) + P_{\lambda_2}^2(\bbeta)$. 

We focus, in theory and practice, on three different choices for the penalty function $P_{\lambda}(\bs\delta)$:
\begin{itemize}
\label{list:penalties}
\item[\textsf{I)}] $P_{\lambda}^H(\bs\delta; \alpha) = \alpha \lambda^2 \sum_{i = 1}^{n+p} \kappa_i^2 \| \delta_i \|_0/2  + (1 - \alpha) \lambda \| \bs\delta \|_2^2 / 2 $, \,
\item[\textsf{II)}] $P_{\lambda}^E(\bs\delta;\alpha) = \alpha \lambda \| \bs\kappa \circ \bs\delta  \|_1   + (1 - \alpha) \lambda \| \bs\delta \|_2^2 / 2$, \, 
\item[\textsf{III)}] $P_{\lambda}^A(\bs\delta;\alpha,\w) =  \alpha\lambda  \|\bs\kappa \circ \w \circ \bs\delta\|_1  + (1 - \alpha) \lambda \| \bs\delta \|_2^2 / 2$. \, 
\end{itemize}
The vector $\bs\kappa = [\kappa_1, \ldots, \kappa_{n+p}]\trans = [k_1\1_n\trans \,\,\, k_2\1_{p}\trans]\trans$
represents the multiplying factors to each model parameter in $\bs\delta, $ $\w$ a vector of non-negative weights, 
and $\alpha \in [0, 1]$ the mixing weight between the sparsity-inducing $\ell_0$/$\ell_1$-norm and the $\ell_2$-norm, 
respectively. The symbol $\circ$ denotes the element-wise product. Penalty function \textsf{I} 
comprises a mixture of the non-convex $\ell_0$ ``norm" and ridge (or Tikhonov) regularization via 
the squared $\ell_2$-norm. Following \citet{she2011outlier}, we refer to this penalty as the 
hard-ridge penalty $P^H$. Penalty \textsf{II} is a convex relaxation of penalty
\textsf{I}, the so-called ``Elastic-Net" penalty \citep{zou2005} 
as mixture of $\ell_1$-norm and squared $\ell_2$-norm and is denoted by $P^E$. Penalty
\textsf{III} augments penalty \textsf{II} by a non-negative weight vector $\w$ in the
$\ell_1$-norm leading to a weighted or ``adaptive" penalty $P^A$, similar to the adaptive lasso
\citep{zou2006}. This convex penalty is novel in the context of mean shift estimation and
requires the construction of appropriate weights $\w$ via a robust data-driven initialization
procedure (see Section \ref{subsec:init}). 

\begin{remark}
The choice of the penalty function determines the properties of the corresponding
robust estimator. In the low-dimensional setting, \citet{Antoniadis2007} and \citet{Gannaz2007}
showed equivalence between Huber's M-estimator and the mean shift model with $\ell_1$ norm
penalization. In robust linear regression, this model, even with added $\ell_2$ regularization
(i.e., the $P^E$ penalty), is prone to ``masking" and ``swapping" effects due to
leveraged outliers \citep{she2011outlier}. Our simulation experiments (see Section \ref{sec:comp})
also confirm this behavior in log-contrast regression. To alleviate this shortcoming
\textsf{RobRegCC} includes the non-convex hard-ridge penalty function $P^H$ \citet{she2011outlier} 
and the convex adaptive penalty function $P^A$ which inherits 
the statistical strength of $P^H$ while simultaneously simplifying computation.
\end{remark}

\section{A unifying computational framework for robust log-contrast regression}
\label{sec:comp}
The computational framework (see Figure~\ref{fig:workflow}) for parameter estimation of
the robust log-contrast regression model in \eqref{eq:cp-Rslr2} comprises three parts:
(i) a novel robust initialization procedure that is instrumental when penalty functions \textsf{I}
or \textsf{III} are used, (ii) a general algorithm for solving the optimization problem in
\eqref{eq:objfun} that can encompass any of the introduced penalty functions, and (iii) a new
robust cross-validation-based (R-CV) model selection strategy specifically tailored to robust
estimation.  

\subsection{A general optimization algorithm}
\label{subsec:optim}
Our algorithmic framework can handle the optimization problem in \eqref{eq:objfun}
with any of the penalty functions \textsf{I}-\textsf{III}. While specialized 
optimization strategies are available for the convex problem instances, (see,
e.g.,\citet{Antoniadis2001,Combettes2011a,she2011outlier,Combettes2012,Briceno-Arias2018}), we present
an general iterative thresholding algorithm, derived from an augmented Lagrangian formulation, that can
encompass all penalty functions. 

A fundamental building block for the proposed algorithm is the use of the proximity or 
thresholding operator $\Theta(\cdot)$ associated with a penalty function $P(\cdot)$:

\begin{align*}
\Theta(t) =  \argmin_{\theta} \frac{1}{2} \| t-\theta\|^2 + P(\theta). 
\end{align*}
For any scalar $a$, the soft thresholding operator is defined 
as $\Theta_\lambda^S(a) = \M{sign}(a)(|a| - \lambda)_+$, 
and the hard threshold operator is $ \Theta_\lambda^H(a) = a  1_{|a| > \lambda}$. 
Table \ref{tab:thresh} summarizes the parameterized scalar thresholding operators, associated with the
penalty functions \textsf{I}-\textsf{III} (see also
\citet{Antoniadis2001,she2011outlier,Combettes2011a}). Note that for vector-valued input to the penalty
functions, the thresholding operators are applied element-wise. 

\begin{table}[htp]
	\centering
	\caption{Penalty function and corresponding thresholding/proximity operator.}\label{tab:thresh}
		\begin{tabular}{lllccccc}
		\hline
		\textsf{Case} & $P_{\lambda}(\theta; \alpha, w$) & $\Theta_{\lambda}(t)$ \\
		\hline
		\mbox{I}& $  \alpha^2 \lambda^2 \kappa^2 \| \theta\|_0/2  + (1 - \alpha) \lambda \| \theta \|_2^2 / 2 $ & $ \Theta_c^{H}(\frac{t}{1+\lambda(1-\alpha)})$ $\text{with}$  $c=\frac{\alpha\lambda\kappa}{\sqrt{1+\lambda(1-\alpha)}}$ \\
		\mbox{II} & $  \alpha \lambda \| \kappa \circ  \theta   \|_1   + (1 - \alpha) \lambda \| \theta \|_2^2 / 2$ & $ \frac{1}{1+\lambda(1-\alpha)}\Theta_{\alpha\lambda\kappa}^{S}(t)$  \\
		\mbox{III} &  $ \alpha\lambda  \| \kappa \circ w \circ\theta \|_1  + (1 - \alpha) \lambda \| \theta\|_2^2 / 2$ & $ \frac{1}{1+\lambda(1-\alpha)}\Theta_{w\alpha\lambda\kappa}^{S}(t)$  \\
		\hline
	\end{tabular}
\end{table}

In the optimization problem in \eqref{eq:objfun}, the linear constraint $\C\trans \bbeta = \0$ implies
$\bbeta = \bP_{\C}^{\perp} \btheta$ for any $\btheta \in \bbR^p$, where $\bP_{\C}^{\perp}$ is the
orthogonal complement of the projection matrix onto subspace $\C$. We can thus reformulate the optimization 
problem in \eqref{eq:objfun} as
\begin{align}
(\what{\btheta}, \what{\bgamma})  \equiv \argmin_{\btheta, \bgamma} \Big\{ f_{\lambda}(\btheta,\bgamma)\Big\} \quad 
\M{s.t.} \,\, \C\trans \btheta = \0 \,, 
\label{eq:objFUNnvar}
\end{align}
where $f_{\lambda}(\btheta,\bgamma) = \frac{1}{2n} \| \y - \X\bP_{\C}^{\perp}\btheta - \sqrt{n}\bgamma
\|_2^2  +   P_{\lambda_1}^1(\bgamma) + P_{\lambda_2}^2(\btheta)$. 
%
%
We solve the constraint optimization problem in \eqref{eq:objFUNnvar} using an augmented Lagrangian approach. The standard augmented Lagrangian for the problem reads
\begin{align*}
L_{\mu , \lambda}( \btheta, \bgamma,\bs\zeta)  =  f_{\lambda}(\btheta, \bgamma) + \bs\zeta\trans\C\trans\btheta  +   \frac{\mu}{2}\|\C\trans\btheta\|_2^2 
\end{align*}
where  $\bs\zeta \in \mathbb{R}^k$ are the Lagrange multipliers and $\mu>0$ is a regularization parameter. 
By reparameterizing $\bs\eta = \bs\zeta/\mu$ and completing the ``square", the augmented Lagrangian simplifies to $L_{\mu , \lambda}( \btheta, \bgamma, \bs\eta)  =  f_{\lambda}(\btheta, \bgamma)  + \frac{1}{2} \|\C\trans\btheta +\bs\eta  \|_2^2$. 

We consider the \textsf{dual descent} approach for solving the associated optimization problem which iterates between 
\begin{itemize}
\item[\textsf{Primal update:}] $(\btheta^{(i+1)}, \bgamma^{(i+1)}) \equiv \argmin_{\btheta, \bgamma} \Bigg\{ L_{\mu , \lambda}( \btheta, \bgamma, \bs\eta^{(i)}) \Bigg\} $,
\item[\textsf{Dual update:}] $\bs\eta^{(i+1)} =  \bs\eta^{(i)} + \C\trans\btheta^{(i+1)}$,
\end{itemize}
until certain convergence criteria are met. The primal update requires solving an unconstrained optimization problem for fixed Lagrange multipliers $\bs\eta^{(i)}$. 
By grouping all terms in the Lagrangian appropriately, we can rewrite this subproblem in standard form  
\begin{align}
\widetilde{\btheta}^{(i+1)} \equiv \argmin_{\widetilde{\btheta}} \,\, \Bigg\{   \frac{1}{2n} \| \widetilde{\y} -\widetilde{\X} \widetilde{\btheta}   \|_2^2 + P_\lambda(\widetilde{\btheta})\Bigg\} \,,
\label{eq:simObjfun}
\end{align}
where 
$$
\widetilde{\y} = \begin{bmatrix}
\y \\ -\sqrt{n}\bs\eta^{(i)}
\end{bmatrix} , \quad  \widetilde{\btheta} = \begin{bmatrix}
\btheta \\ \bgamma
\end{bmatrix} ,  \quad \M{and} \quad
\widetilde{\X} = \begin{bmatrix}
\X\bP_{\C}^{\perp} & \sqrt{n}\I \\ \sqrt{n}\C\trans & \0
\end{bmatrix}. 
$$
For the penalty functions \textsf{I}-\textsf{III}, this problem formulation is amenable to iterative 
shrinkage/thresholding algorithms (ISTA) (see, e.g., \cite{Daubechies2004,Combettes2011a}) or,
equivalently, to the thresholding-based iterative selection procedure (TISP) \citep{She2009}.
Convergence guarantees, however, depend on the specific properties of the penalty function. ISTA
algorithms comprise a (forward) gradient step and a (backward) proximal/thresholding step. To solve the
primal update at the $(i+1)$th stage for the objective in \eqref{eq:simObjfun}, the $(j+1)$th iteration
in ISTA reads
\begin{align}
\widetilde{\btheta}^{(i,j+1)} = \Theta_{\lambda}\Bigg[  \widetilde{\btheta}^{(i,j)} - \frac{1}{n k_0} 
\widetilde{\X}\trans \big(\widetilde{\y} - \widetilde{\X}\widetilde{\btheta}^{(i,j)}\big)\Bigg] \,,
\label{eq:inLoop}
\end{align}
where $\Theta_\lambda[\cdot]$ is the thresholding operator corresponding to the considered 
penalty function (see Table \ref{tab:thresh}). The operator is applied element-wise to the 
entries of the vectors. The
iterative algorithm is stopped when a prescribed convergence criterion on the consecutive iterates is
reached. To ensure monotone decrease in the objective function, the scaling constant needs to satisfy 
$k_0<\frac{1}{2}\sigma_{\widetilde{\X}}$ where $\sigma_{\widetilde{\X}}$ is largest eigenvalue of 
$\widetilde{\X}\trans\widetilde{\X}$ (see, e.g,
\citet{She2009,Bayram2016a}). Global and local convergence of the iterates can be proven for convex and
non-convex penalties, respectively \citep{Bauschke2011,Bayram2016a,She2009}. 
In order to solve the primal update fast and robustly, 
we provide penalty-dependent initial parameter estimates $\widetilde{\btheta}^{(i,0)}$. 
For the convex penalties \textsf{II-III}, we employ a ``warm start" strategy and set 
$\widetilde{\btheta}^{(i,0)} =\widetilde{\btheta}^{(i-1)}$. When the non-convex penalty \textsf{I} is
used, we set $\widetilde{\btheta}^{(i,0)} = \ddot{\bdelta}$, which is the solution of our robust
initialization procedure, detailed in Section \ref{subsec:init}. This robust solution is also used to
construct weights for penalty \textsf{II}. Following \citet{zou2006}, we set $\w = |\ddot{\bdelta}|^{-\nu}$ with $\nu=1$.
For penalties \textsf{I} and \textsf{III}, the weights are set to $\w = \1_{n+p}$. 

The convergence of the dual descent approach naturally depends on the convergence of the algorithms
solving the subproblems. We refer to \citet{Bertsekas1982}, Proposition 2.1 - 2.3, for 
convergence guarantees regarding the Lagrangian approach. With default parameters $\alpha = 0.95$
and $\nu = 1$, we see fast and robust performance of the algorithm under a wide range of simulation
and application scenarios. We summarize the estimation procedure in Algorithm \ref{alg:rbest}.

\begin{algorithm}[!h]
\caption{Robust regression with compositional covariates} 
\begin{algorithmic}\label{alg:rbest}
\STATE Given: $\X,\, \y, \, \lambda >0, \, \nu=1, \, \alpha = 0.95$.
\STATE Choose penalty function $\M{P}_{\lambda}(\cdot)$ and  threshold operator $\Theta_{\lambda}(\cdot)$ from Table \ref{tab:thresh}. 
\STATE Set i=0, $ \widetilde{\btheta}^{(0)} =  \ddot{\bs \delta} $ (initialization) and  compute the scaling constant $k_0$. 
\REPEAT 
\STATE Define $\{\widetilde{\X}, \widetilde{y}\}$ using \eqref{eq:simObjfun}, and set  $j=0$; $\widetilde{\btheta}^{(i,0)} = \widetilde{\btheta}^{(i)}$ (\textsf{case II and III});  $\widetilde{\btheta}^{(i,0)} = \ddot{\bs \delta}$ (\textsf{case I}) 
\REPEAT
    \STATE $
\widetilde{\btheta}^{(i,j+1)} = \Theta_{\lambda}\Bigg[  \frac{\widetilde{\X}\trans \widetilde{\y}}{k_0} +\Bigg( \I - \frac{\widetilde{\X}\trans \widetilde{\X}}{k_0} \Bigg) \widetilde{\btheta}^{(i,j)} \Bigg]
$
\STATE $j\gets j+1$. 
\UNTIL{convergence} 
\RETURN $\widetilde{\btheta}^{(
i+1)}$
\vspace{0.2cm}
\STATE $\eta^{(i+1)} =  \eta^{(i)} + \C\trans {\btheta}^{(i+1)}$ and then $i\gets i+1$.
\UNTIL{convergence}
\end{algorithmic}
\end{algorithm}

\subsection{Robust initialization}
\label{subsec:init}
Robust estimation procedures for linear models, including S-estimators \citep{rousseeuw1984robust}, MM-estimators
\citep{Yohai1987}, the $\Theta$-IPOD \citep{she2011outlier}, and the Penalized Elastic-Net S-Estimator (PENSE) 
\citep{CohenFreue2017}, are multi-stage estimators that comprise an initialization stage and one or several
improvement stages. A common theme for the initialization stage is the use of resampling-based approaches in
combination with robust loss functions \citep{Maronna2006,salibian2006fast}. While most methods assume the model
coefficients to be dense, a variant of the $\Theta$-IPOD as well as the PENSE encourage sparse coefficients.
However, the latter methods operate under the standard linear model and are not suited for the compositional
setting. When solving the optimization problem in \eqref{eq:objFUNnvar} with penalty \textsf{I} or \textsf{III},
RobRegCC requires an initial estimate of the coefficients and the mean shift vector. Here, we apply the concept
of principal sensitivity components (PSC) \citep{pena1999fast} to the log-contrast model and propose the
following resampling-based approach to robust initialization. 

In the linear model, PSC analysis has been introduced for ordinary least squares \citep{pena1999fast} 
and extended to robust ridge regression \citep{maronna2011robust} and robust sparse regression
\citep{CohenFreue2017}. PSC analysis relies on the idea of leave-one-out sensitivities of the following form.
Given all samples, let $\what{y}_i$ be the estimated prediction of the model under consideration for observation
$i$, and $\what{y}_{i(j)}$ the corresponding prediction value with the $j$th observation removed. The sensitivity
of the $i$th observation is then defined as
\begin{align}
\r_{i} = [\what{y}_i - \what{y}_{i(1)},\ldots,\what{y}_i - \what{y}_{i(n)}]\trans , \qquad i=1,\ldots,n . 
\label{eq:psc}
\end{align}
The sensitivity matrix of all observation is defined as $\R = [\r_{1},\ldots,\r_{n}] \in \mathbb{R}^{n \times n}$. 
To identify potential outliers in the data, PSC analysis proceeds by first computing 
an eigenvalue decomposition of the matrix $\R \R^\trans$. The eigenvectors $\u_i$ of that matrix are called 
the principal sensitivity components (PSCs) of the matrix $\R$ \citep{pena1999fast}. Observations that comprise 
extreme values with respect to the PSCs are deemed outliers and removed from the samples. 

In \textsf{RobRegCC}, we adopt the protocol of PENSE (see \citep{CohenFreue2017} 2.2.) and propose a PSC-based
analysis on the standard sparse log-contrast model \citep{shi2016regression}, resulting in initial coefficient
estimates on potentially outlier-free subsamples. We provide the exact computational protocol in Section \ref{suppl:robust} of the Supplementary Materials. 
The final outcome of the robust initialization procedure is the estimate $\ddot{\bdelta} = [\ddot{\bbeta}\trans \,
\ddot{\bgamma}\trans]\trans$ which serves as initial starting 
point of the non-convex optimization procedure underlying \textsf{RobRegCC} with penalty $P^H$ and forms the basis
for weight construction in the adaptive Elastic Net penalty $P^A$, respectively.    

\subsection{Robust cross-validation model selection}
\label{subsec:mdlSel}
An essential part of the \textsf{RobRegCC} workflow is data-driven tuning of the regularization parameter
$\lambda$. Due to the lack of model selection criteria for the log-contrast model with corrupted
observations, we introduce a novel robust cross-validation (R-CV) strategy. 

We consider a $\lambda$-path with log-linearly spaced $\lambda$ values in the interval
$[\lambda_\text{min},\lambda_\text{max}]$. We set the upper bound $\lambda_\text{max} =
\M{max}(|[\y\,\,\, \X\trans\y ]|/\w)$ and $\lambda_\text{min}$ to be the fraction of
$\lambda_\text{max}$ at which the mean shift parameter $\hat \bgamma$ comprises at most $n/2$ non-zeros
(i.e., potential outliers). We split the model selection training data into $k = 10$ folds and perform
cross-validation with a specifically tailored robust test statistics. For a given fold, we denote the $n_{tr}$ training 
data by $\{\y_{tr}, \X_{tr}\}$, the
$n_{te}$ test data by $\{\y_{te}, \X_{te}\}$, and the parameters estimates on the training data at a given $\lambda$ by
$\{\what{{\btheta}}_{\lambda},\what{\bgamma}_{\lambda}\}$. As the standard mean-squared error is not an appropriate error
measure in robust regression, we first compute the robust scale estimate $\hat s_{tr,\lambda}$ on the training data using
\[
\hat s_{tr,\lambda} = \sqrt{\|\hat \bepsilon_{tr,\lambda}\|^2/n_{tr}}\\ 
\quad \text{with} \quad \hat \bepsilon_{tr,\lambda} = \y_{tr} - \X_{tr}\what{\btheta}_{\lambda} - \what{\bgamma}_{\lambda} \,.
\]
The non-robust test sample residual is $\hat \bepsilon_{te,\lambda} = \y_{te} - \X_{te}\what{\btheta}_{\lambda}$. 
To account for unknown outliers in a test sample, we calculate the scaled test error
\[
\hat \r_{te,\lambda} = \hat \bepsilon_{te,\lambda}/\hat s_{tr,\lambda}
\]
and then derive the robust scale estimate for the test sample using the median absolute deviation (MAD) \citep{rousseeuw2011robust}:
\[
\hat{s}_{te,\lambda} = \text{MAD}(\hat \r_{te,\lambda})\,. 
\]
This scale estimate enables the removal of outliers in the test sample. Let 
$r_{te,j,\lambda}$ be the residual of the $j$th test sample. 
The set of outlier in the test sample $O_{te}$ is 
\[
 O_{te} = \{j \in \{1,\ldots, n_{te}\} \, | \quad |r_{te,j,\lambda} | > 2 \hat{s}_{te,\lambda} \}. 
\]
After removing outliers $O_{te}$ from the test sample, we denote the ``clean" 
test data by $\{\y^c_{te}, \X^c_{te}\}$ and the ``clean" residual by $\hat \bepsilon^c_{te,\lambda} = \y^c_{te} -
\X^c_{te}\what{\btheta}_{\lambda}$. We calculate the standard deviation $\hat \sigma^c_{te,\lambda}$ of the residual 
and introduce the following test statistics for robust cross-validation:
\begin{equation}
\hat \epsilon^c_{te,\lambda} = |\hat \sigma^c_{te,\lambda} -1|
\label{eq:teststats}
\end{equation}
Let $\hat \epsilon^c_{te,i,\lambda}$ be the robust test statistics for fold $i$. We select the tuning parameter 
$\lambda_\text{r-cv}$ that minimizes the average k-fold robust cross-validation error $\bar \epsilon^c_{te,\lambda} = 
\frac{1}{k} \sum_{i=1}^{k} \hat \epsilon^c_{te,i,\lambda}$. 

Figure \ref{fig:modelfitdiag} illustrates the typical behavior of R-CV model selection over the $\lambda$-path for
simulated data (see Section \ref{sec:siml} for details). 

\begin{figure}[htp]
\begin{center}
\includegraphics[width=1\textwidth, angle=0]{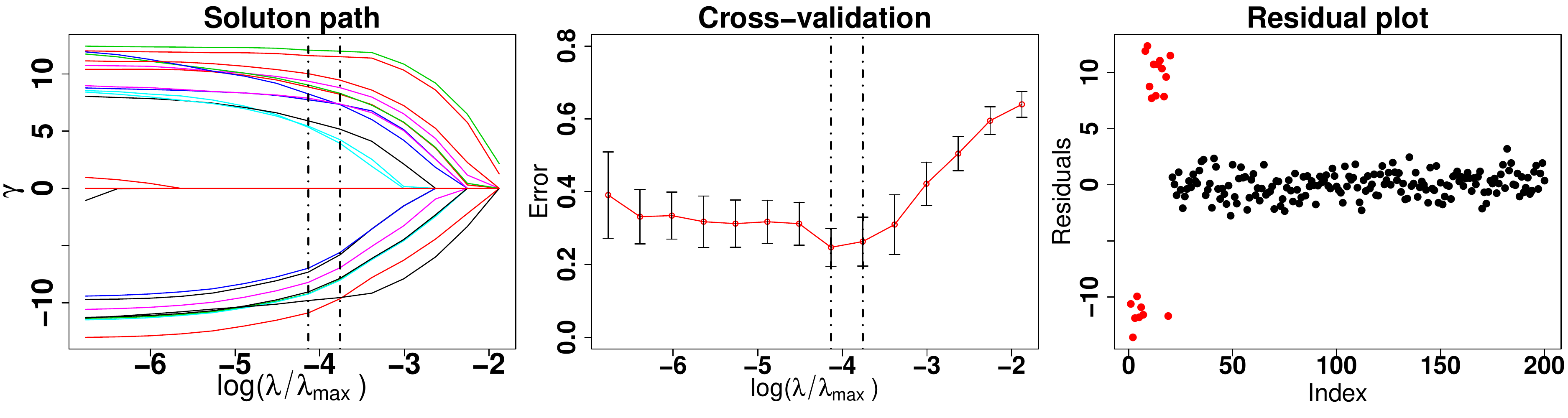}
\caption{\textsf{R-CV model selection}: Left panel: Solution path of the mean shift parameter estimates $\hat \bgamma$;
middle panel: R-CV error $\bar \epsilon^c_{te,\lambda}$ (and standard error over k=10 folds) across the $\lambda$-path
using \textsf{RobRegCC} with adaptive Elastic-net penalty. Vertical dashed line corresponds to the minimum R-CV error
(left) and the one standard error (1SE) rule; right panel: Residuals $\hat \bepsilon_{te,\lambda_\text{r-cv}}$ with
outliers (red)  identified using \textsf{RobRegCC} at minimum R-CV error.}\label{fig:modelfitdiag}
\end{center}
\end{figure}

\begin{remark}
\label{remark:ote}
We highlight that the robust test statistics in \eqref{eq:teststats} is also a useful measure for judging the 
performance of out-of-sample prediction, i.e., testing predictive power of an estimator on a hold-out (or validation)
set after model selection. We will use this measure in simulation and real data analysis in Sections~\ref{sec:siml} and \ref{sec:application}. 
\end{remark}

\section{Non-asymptotic Analysis}\label{sec:nasyanal}
We observe that the number of unknown parameters in the robust log-contrast regression model  
increases linearly with the sample size $n$.
Hence, a finite sample analysis is desirable to understand the effect of the number of samples $n$, 
number of predictors $p$, and linear constraints $k$ on the model prediction error. 
For simplicity, we perform the analysis of the robust model  \eqref{eq:cp-Rslr2} with 
only compositional covariates $\Z$, i.e, 
\begin{align}
\y =  \Z\bbeta^*+ \bgamma^* + \bepsilon, \qquad \M{s.t.} \qquad \C\trans \bbeta^*= \0, \label{eq:trumdl}
\end{align}
where $\bbeta^* \in \mathbb{R}^p$ is the true coefficient, $\bgamma^* \in \mathbb{R}^n$ is the true mean shift,
and $\bepsilon$ is the IID sub-Gaussian error with mean zero and variance $\sigma^2$. $\T^* = \bs\mcJ(\bgamma^*)$
and  $\S^* = \bs\mcJ(\bbeta^*)$ denote the support index sets of $\bgamma^*$ and $\bbeta^*$ such that $|\T^*| = t^*$
and $|\S^*| = s^*$, respectively.

From the general RobRegCC model formulation \eqref{eq:objfun21}, the optimization problem for the reduced model \eqref{eq:trumdl} is given by 
\begin{align}
\{\what{\bgamma}_{\lambda_1},\what{\bbeta}_{\lambda_2}\} \equiv \argmin_{\{\bgamma,\bbeta\}} \Bigg\{ \frac{1}{2} \| \y
- \Z \bbeta - \bgamma \|_2^2 & + P_{\lambda_1}^1(\bgamma) + P_{\lambda_2}^2(\bbeta) \Bigg\}   
 \quad \M{s.t.} \quad \C\trans \bbeta = \0 \,. \label{eq:objfun21x}  
\end{align}
In order to focus on the core issue, we consider dropping the quadratic component of the penalty functions (equivalent
to setting $\alpha = 1$ as defined in Table \ref{tab:thresh}) with separate tuning parameters $\{\lambda_1, \lambda_2\}$
for penalizing $\{\bgamma, \bbeta\}$, respectively. Hereafter, for the ease of notation, we drop the subscript and
denote the optimal solution of \eqref{eq:objfun21x} by  $\{\what{\bbeta}, \what{\gamma}\}$. 

We define the model prediction error as $\M{M}(\what{\bbeta} - \bbeta^*, \what{\bgamma} - \bgamma^*)$ where $\M{M}(\a,\b) = \| \Z\a + \b \|_2^2 $. Our analysis upper-bounds the prediction error in terms of the model bias and
variance. 
For the unrestricted predictor matrix $\Z$, Theorem \ref{sec:non-asy:th1} provides a \textsf{slow rate bound} on the
prediction error regardless of the type of sparsity-inducing penalty functions. Consequently, remark following the
theorem states the oracle bound in case of $\ell^0$ norm penalty (\textsf{case I}). Theorem \ref{sec:non-asy:th2} 
provides the result to attain the required oracle bound in case of $\ell^1$ penalty under a compatibility condition on
$\Z$, also referred as \textsf{fast rate bound}. Moreover, under some
additional regularity assumptions, the finite sample analysis of the prediction error can be extended i) to perform the
asymptotic analysis; ii) to obtain the estimation error bound in various norms; and iii) to establish selection
consistency of the parameter estimates \citep{lounici2011oracle}. 

The proofs of our theorems rely on and extend prior work, in particular
\citet{Shea,She2017b,she2017selective}. 
\begin{theorem}\label{sec:non-asy:th1}
Consider the tuning parameter $\lambda_1 = A\lambda_a$  and $\lambda_2 = A\lambda_b$ with $\lambda_a =
\sigma(\log(en))^{1/2}$, $\lambda_b = \sigma(\log(ep))^{1/2}$, and $A = \sqrt{ab}A_1$  for a sufficiently large
$A_1$ satisfying $a \geq 2b > 0$.   In terms of the optimal solution $\{\what{\bbeta}, \what{\bgamma}\}$ of the
optimization problem \eqref{eq:objfun21x}, we have 
\begin{align*}
M(\what{\bbeta}-\bbeta^*,  \what{\bgamma} - \bgamma^*) \,\, \lesssim \,\, M(\bbeta-\bbeta^*,  \bgamma -
\bgamma^*)   &  + P_{\lambda_1}^1(\bgamma) + P_{\lambda_2}^2(\bbeta) +  (3-k)\sigma^2, 
\end{align*}
for any $\{{\bbeta}, {\bgamma}\}$. 
Here $\lesssim$ means the inequality holds up to a multiplicative constant. 
\end{theorem}

\begin{remark}\label{remark:1}
Consider $P_{\lambda_1}^1(\bgamma) \lesssim \lambda_1^2$ and $P_{\lambda_2}^2(\bbeta) \lesssim \lambda_2^2$. In
case of $\bbeta = \bbeta*$ and $\bgamma = \bgamma*$,  it follows from Theorem \ref{sec:non-asy:th1} that 
\begin{align}
M(\what{\bbeta}-\bbeta^*,  \what{\bgamma} - \bgamma^*) \,\, &  \lesssim \,\, (3-k)\sigma^2 + P_{\lambda_1}^1(\bgamma^*) + P_{\lambda_2}^2(\bbeta^*) 
 \lesssim \,\, (3-k)\sigma^2 + \frac{\lambda_1^2}{2}  t^* + \frac{\lambda_2^2}{2}  s^*. \label{theres:oracle}
\end{align}
The oracle bound suggests that, with moderate number of outlier, the dependence of variance on $t^*$ allows the 
parameter estimate to reduce model bias. 
\end{remark}

\begin{corollary}\label{corollary:2}
Let us assume $0 \log 0 = 0$. Following the proofs of Theorem \ref{sec:non-asy:th1} and  Lemma \ref{lemma:lm1}
in the SM, we have 
\begin{align*}
M(\what{\bbeta}-\bbeta^*,  \what{\bgamma} - \bgamma^*) \, \lesssim \, \inf_{(\bbeta  , \bgamma ;\, t \leq \vartheta )} \, M(\bbeta & -\bbeta^*,  \bgamma - \bgamma^*)   + \sigma^2 + P_{\lambda_2}^2(\bbeta) + 2a L \sigma^2 \Big(  \vartheta  + \vartheta  \log \frac{en}{\vartheta } +2 -k \Big),
\end{align*}
where $t$ is cardinality of $\bs\mcJ(\bgamma)$.
\end{corollary}
\begin{remark}\label{remark:2}
Let us assume that $\y$ is obtained after corrupting  $\vartheta$ outcomes in the true generating model $\y^* =
\Z\bbeta^* + \bepsilon$.
We define the breakdown point of the robust model as  $\epsilon^*(\what{\bbeta},\what{\bgamma}) =
\min\{\vartheta/n\,; \sup_{ |\bs\mcJ(\bgamma^*)| \leq {\vartheta}} \mathbb{E}[M(\what{\bbeta}-\bbeta^*, 
\what{\bgamma} - \bgamma^*) ] = \infty\}$.  From the  Corollary \ref{corollary:2}, it follows that the finite
sample breakdown point of the robust model is given by  $\epsilon^* \geq (\vartheta+1)/n$.
\end{remark}
%
%
The parameter estimate obtained by solving the optimization problem \eqref{eq:objfun21x} with $\ell^1$ penalty attains the required oracle bound under the following compatibility conditions: 
\begin{itemize}
\item[\textbf{C1.}] $
(1+\nu) |\bgamma_{\T}^'|_1 \leq |\bgamma_{\T^c}^'|_1 + \kappa_1 t^{1/2} \|\bP_{\bs\zeta}^{\perp}\bgamma^' \|_2
$
\item[\textbf{C2.}] $
(1+\nu) |\bbeta_{\S}^'|_1 \leq |\bbeta_{\S^c}^'|_1 + \kappa_2 s^{1/2} \|\bP_{\bs\zeta} ( \bs\zeta\bbeta^{''}+ \bgamma^') \|_2
$
\end{itemize}
for any suitable dimension $\bgamma^'$, $\bbeta^'$ $\bbeta^{''}$, and the projection matrix $\bP_{\bs\zeta}$ 
mapping the column space of $\bs\zeta \subseteq\Z$. Here, parameters $\kappa_1$, $\kappa_2$ and $\nu$ are
positive compatibility constants  . 

\begin{theorem}\label{sec:non-asy:th2}
Consider the model matrix $\Z$ in \eqref{eq:trumdl} satisfies the compatibility condition $\{\textbf{C1}, \textbf{C2}\}$,  and the tuning
parameter $\lambda_1 = A\lambda_a$  and $\lambda_2 = A\lambda_b$ with $\lambda_a =
\sigma(\log(en))^{1/2}$, $\lambda_b = \sigma(\log(ep))^{1/2}$, and $A = \sqrt{ab}A_1$ for a sufficiently large
$A_1$ satisfying $a \geq 2b > 0$. In terms of the optimal solution $\{\what{\bbeta}, \what{\bgamma}\}$ of the
optimization problem \eqref{eq:objfun21x} with $\ell_1$ penalty  i.e., $P_{\lambda_1}^1(\bgamma) =
\lambda_1|\bgamma|_1$ and $P_{\lambda_2}^2(\bbeta) = \lambda_2|\bbeta|_1$, we have 
\begin{align*}
 M(\what{\bbeta}-\bbeta^*,  \what{\bgamma} - \bgamma^*) \lesssim    M(\bbeta-\bbeta^*,  \bgamma - \bgamma^*) & + a\lambda_1^2 (1-\theta)^2 \kappa_1^2 t +     a\lambda_2^2 (1-\theta)^2 \kappa_2^2 s  + (3-k)\sigma^2
\end{align*}
where $\theta = \nu/ (1 + \nu)$, $t = |\bs\mcJ(\bgamma)|$ and $s = |\bs\mcJ(\bbeta)|$. 

\end{theorem}

\begin{remark}
Compared to the oracle bound in Equation \eqref{theres:oracle}, the variance term of the prediction error
bound in Theorem \ref{sec:non-asy:th2} differs only by a constant multiplying factor.
\end{remark}

\section{Simulation benchmarks}
\label{sec:siml}
The overall purpose of the following simulation study is to evaluate \textsf{RobRegCC}'s ability to
simultaneously detect outliers and to perform sparse covariate selection when the underlying generative
model is sparse. We follow the original simulation setup for the standard log-contrast model, put forward
in \citet{shi2016regression}, and extend it by introducing different types of outliers in the response. 
We remark that the synthetic simulation setup does not reflect all aspects of high-throughput sequencing
count data. Our simulations will be complemented by real gut microbiome data analysis in Section
\ref{sec:application}.

\subsection{Benchmark setup}
\label{subsec:setup}
Following the simulation setup in \citet{shi2016regression}, 
we generate count data $\W  = [w_{ij}] \in \bbR^{n \times p}$ by simulating $n$ instances of a multivariate
random variable $\w \sim \M{Lognormal}(\bs\mu,\bs\Sigma)$ with mean $\bs\mu = \{
\log(p/2) \, \1_5,  \,\0_{p-5} \} \in \bbR^p$ and covariance matrix $\bs\Sigma  \in \bbR^{ p \times p}$ such
that $\Sigma_{ij} = 0.5^{|i-j|}$. 
We perform total-sum normalization of the count data $\W$ and apply a log-transformation on the resulting
compositions, thus arriving at covariates  ${\Z} = [{z}_{ij}]_{n \times p}$ where $z_{ij} = \log
\frac{w_{ij}}{\sum_{j=1}^p w_{ij}}$. For simplicity, we include a single non-compositional covariate in the
form of an intercept $\bN = \1_n$. Using the generative model in \eqref{eq:cp-slr}, we define $\X = [\bN \,\,
\Z]$ and set $\bbeta^* = \{\beta_0, 1, - 0.8, 0.4, 0, 0, - 0.6, 0, 0, 0, 0, -1.5, 0, 1.2, 0, 0, 0.3,
\0_{p-16}\}$ with $\beta_0 = 0.5$ \citep{shi2016regression}.  
To model sub-compositional coherence, we consider the subcomposition constraint matrix $\C$ with $k=4$
subgroups of the form
\begin{align}
\bar{\C}\trans = \begin{bmatrix}
    \1_{p_1}\trans & \0 & \dots  & \0 \\
    \0  & \1_{p_2}\trans &  \dots  & \0 \\
    \vdots & \vdots &  \ddots & \vdots \\
    \0 & \0 &  \dots  & \1_{p_4}\trans 
\end{bmatrix}_{4 \times 23} \M{and} \qquad 
\C\trans = \begin{bmatrix}\0_{4 \times 1} & \bar{\C}\trans & \0_{4 \times \{p-23\}} \end{bmatrix}_{4 \times (p+1)}
\label{eq:defCmat2}
\end{align}
with $\s = [0,10, 16, 20, 23]$, $p_i = \s_{i+1}-\s_{i}, i=1,\ldots,8$, 
and index sets $\mathbb{A}_{i} = \{\s_{i}+1,\ldots,\s_{i+1}\}$. 

We first generate the outlier-free response $\y$ with error standard deviation $\sigma =
\|\widetilde{\Z}\bbeta\|/(\sqrt{n} \times \M{SNR})$ where the signal-to-noise ratio (SNR) is set to
$\M{SNR}=3$. For fixed $n = 200$, we examine both the low- and high-dimensional scenario using $p = \{100, 300\}$.

To evaluate the ability of \textsf{RobRegCC} to detect outliers in the response, we considered the following
scenarios for outlier generation. We used mean shift vectors $\bgamma$ 
with $\bO = \{ 20\%, 15\%, 10\%, 5\%\} n$  outliers. 
\textsf{Moderate} outliers were generated by adding a shift of 6$\sigma$ to the 
true response $\y$, and \textsf{large} outliers by adding a shift of 8$\sigma$ to the true response $\y$. 
We also considered the challenging setting where half of the $\bO$ outliers are leveraged. 
Leveraged instances in ${\Z}$ (denoted by L in the simulation scenarios) 
were obtained by modifying the entries in the count data $\W$. The first $\bO/2$ instances of $\W$ 
are replaced by leveraged observations. A leveraged observation comprises a covariate (taxon) in 
each subgroup that is inflated to a large value while the remaining taxa in the subgroup are 
deflated to small values. For each subgroup of the $k$ groups, we first identified the 
corresponding column subset matrix of $\W$, then arranged its first column in descending order after 
adding the constant $4$, and then appended the remaining columns in ascending order. 
The first $\bO/2$ instances of the rearranged matrix were the leveraged observations. We replicated 
each experimental setting $R=100$ times.

For outlier identification, we measured performance in terms of the number of  
false positive (\textsf{FP}) (``swapping") and false negative (\textsf{FN}) (``masking") outliers. 
For the standard \textsf{RobRegCC} workflow, the number of \textsf{FP} and \textsf{FN} are derived 
by comparing the true set of outliers to the support set $\bs\mcJ(\hat{\bgamma}_{\lambda_{\text{r-cv}}})$
after R-CV model selection. We denote these estimates by $\M{FP}_1$ and $\M{FN}$. 
We also provide a two-stage estimator, where we refit the standard log-contrast model
on the identified inliers, compute the standard deviation of the residuals, and redefine all samples as
inliers if their residuals are within the range of three standard deviations. The number of false positives
for the two-stage estimator is denoted by $\M{FP}_2$. Total mis-identification performance is measured
using the Hamming distance, HM = FN+$\M{FP}_1$. We determine the quality of 
\textsf{RobRegCC}'s estimated sparse regression coefficients by refitting a standard log-contrast model 
on the support of $\bs\mcJ(\hat{\bbeta}_\lambda)$ at $\lambda=\lambda_\text{r-cv}$ 
using the inlier data only. The resulting refit $\hat \bbeta_\text{rf}$ estimates are compared 
to the oracle $\bbeta^*$ via the scaled estimation error $\M{Er}({\bbeta}) = 100\|\bbeta^* - \hat \bbeta_\text{rf}\|_2/p$.

\subsection{Simulation Results}
We summarize \textsf{RobRegCC}'s performance in the setting with \textsf{large} outliers (shift $s = 8\sigma$) 
in Table \ref{tab:sequal8}. Similar results for the \textsf{moderate} outlier (shift $s = 6\sigma$) scenario are
available in Table \ref{tab:sequal6} of the Supplementary Material. For comparison, we also consider the 
standard log-contrast model without mean shift (denoted by NR).

\begin{table}[H]

\caption{\label{tab:sequal8}Comparison of the non-robust [NR] model and the robust model, i.e.,  RobRegCC with the hard-ridge [H], the Elastic Net [E], and the adaptive Elastic Net [A] penalty function, in the simulation setting with high outliers (n = 200, s = 8) using the outlier identification measures false negative (FN) and false positive (FP), and the estimation error measure Er($\bbeta$). Here $\M{FP}_1$ and $\M{FP}_2$ refers to the pre and post false positive measures.\\}
\centering
\resizebox{\linewidth}{!}{
\fontsize{7.5}{9.5}\selectfont
\begin{tabu} to \linewidth {>{\raggedleft}X>{\raggedleft}X>{\raggedleft}X>{\raggedright}X>{\raggedright}X>{\raggedright}X>{\bfseries}l>{\raggedright}X>{\raggedright}X>{\raggedright}X>{\bfseries}l>{\raggedright}X>{\raggedright}X>{\raggedright}X>{\bfseries}l>{\bfseries}l}
\toprule
\multicolumn{3}{c}{\textbf{ }} & \multicolumn{4}{c}{\textbf{[A]}} & \multicolumn{4}{c}{\textbf{[H]}} & \multicolumn{4}{c}{\textbf{[E]}} & \multicolumn{1}{c}{\textbf{[NR]}} \\
\cmidrule(l{3pt}r{3pt}){4-7} \cmidrule(l{3pt}r{3pt}){8-11} \cmidrule(l{3pt}r{3pt}){12-15} \cmidrule(l{3pt}r{3pt}){16-16}
\multicolumn{3}{c}{ } & \multicolumn{3}{c}{$\bgamma$} & \multicolumn{1}{c}{$\bbeta$} & \multicolumn{3}{c}{$\bgamma$} & \multicolumn{1}{c}{$\bbeta$} & \multicolumn{3}{c}{$\bgamma$} & \multicolumn{1}{c}{$\bbeta$} & \multicolumn{1}{c}{$\bbeta$} \\
\cmidrule(l{3pt}r{3pt}){4-6} \cmidrule(l{3pt}r{3pt}){7-7} \cmidrule(l{3pt}r{3pt}){8-10} \cmidrule(l{3pt}r{3pt}){11-11} \cmidrule(l{3pt}r{3pt}){12-14} \cmidrule(l{3pt}r{3pt}){15-15} \cmidrule(l{3pt}r{3pt}){16-16}
\rowcolor{gray!6}  L & p & O & FN & $\M{FP}_1$ & $\M{FP}_2$ & Er & FN & $\M{FP}_1$ & $\M{FP}_2$ & Er & FN & $\M{FP}_1$ & $\M{FP}_2$ & Er & Er\\
\midrule
0 & 100 & 0 & 0.00 & 3.25 & 1.47 & 1.11 & 0.00 & 5.25 & 1.82 & 1.12 & 0.00 & 9.63 & 2.35 & 1.13 & 1.04\\
\rowcolor{gray!6}  0 & 100 & 10 & 0.00 & 1.68 & 1.04 & 1.11 & 0.00 & 0.82 & 0.70 & 1.12 & 0.00 & 13.82 & 3.78 & 1.17 & 1.84\\
0 & 100 & 20 & 0.00 & 0.98 & 0.75 & 1.09 & 0.27 & 0.00 & 0.42 & 1.09 & 0.00 & 14.15 & 3.26 & 1.16 & 2.30\\
\rowcolor{gray!6}  0 & 100 & 30 & 0.00 & 0.40 & 0.54 & 1.14 & 0.42 & 0.00 & 0.35 & 1.12 & 0.00 & 16.86 & 3.47 & 1.19 & 2.65\\
0 & 100 & 40 & 0.00 & 0.00 & 0.40 & 1.16 & 1.33 & 0.00 & 0.00 & 1.18 & 0.00 & 17.62 & 3.60 & 1.23 & 3.32\\
\hline
\addlinespace
\rowcolor{gray!6}  0 & 300 & 0 & 0.00 & 4.44 & 1.78 & 0.49 & 0.00 & 7.12 & 1.88 & 0.47 & 0.00 & 9.39 & 2.26 & 0.49 & 0.47\\
0 & 300 & 10 & 0.00 & 2.28 & 1.29 & 0.50 & 0.00 & 2.11 & 1.07 & 0.50 & 0.00 & 9.47 & 2.45 & 0.51 & 0.89\\
\rowcolor{gray!6}  0 & 300 & 20 & 0.00 & 1.00 & 0.99 & 0.50 & 0.00 & 0.41 & 0.53 & 0.50 & 0.00 & 10.70 & 2.64 & 0.52 & 1.17\\
0 & 300 & 30 & 0.00 & 0.37 & 0.46 & 0.55 & 0.35 & 0.00 & 0.36 & 0.55 & 0.00 & 11.67 & 2.72 & 0.57 & 1.39\\
\rowcolor{gray!6}  0 & 300 & 40 & 0.00 & 0.45 & 0.55 & 0.55 & 1.10 & 0.00 & 0.23 & 0.55 & 0.00 & 11.88 & 2.27 & 0.60 & 1.39\\
\hline
\addlinespace
1 & 100 & 10 & 0.00 & 1.37 & 1.03 & 1.06 & 0.00 & 0.86 & 0.60 & 1.07 & 0.34 & 11.01 & 2.52 & 1.16 & 2.57\\
\rowcolor{gray!6}  1 & 100 & 20 & 0.00 & 0.42 & 0.53 & 1.11 & 0.00 & 0.00 & 0.45 & 1.13 & 2.94 & 11.74 & 2.59 & 1.47 & 2.37\\
1 & 100 & 30 & 0.00 & 0.40 & 0.46 & 1.17 & 1.03 & 0.00 & 0.30 & 1.18 & 6.84 & 10.05 & 2.47 & 1.89 & 2.62\\
\rowcolor{gray!6}  1 & 100 & 40 & 6.25 & 0.13 & 0.39 & 1.76 & 7.82 & 0.00 & 0.00 & 1.61 & 11.47 & 9.98 & 2.23 & 2.32 & 2.86\\
\hline
1 & 300 & 10 & 0.00 & 1.74 & 1.14 & 0.54 & 0.00 & 1.54 & 1.12 & 0.53 & 0.43 & 9.60 & 2.48 & 0.64 & 1.01\\
\addlinespace
\rowcolor{gray!6}  1 & 300 & 20 & 0.00 & 0.99 & 0.47 & 0.55 & 0.00 & 0.00 & 0.30 & 0.56 & 6.45 & 6.21 & 1.55 & 1.14 & 1.15\\
1 & 300 & 30 & 0.00 & 0.84 & 0.60 & 0.55 & 0.41 & 0.00 & 0.36 & 0.52 & 11.26 & 6.20 & 1.76 & 1.16 & 1.12\\
\rowcolor{gray!6}  1 & 300 & 40 & 1.44 & 0.38 & 0.35 & 0.68 & 0.99 & 0.00 & 0.28 & 0.61 & 16.64 & 7.93 & 1.63 & 1.03 & 1.05\\
\bottomrule
\end{tabu}}
\end{table}
We observed that \textsf{RobRegCC} with hard-ridge(H) and adaptive penalty(A) consistently outperformed the other
methods both in terms of outliers identification and regression coefficient estimation. 
This highlights the importance of our novel PSC-based initialization routine on subsequent estimation.
\textsf{RobRegCC} with the Elastic Net penalty (E) performs well in the absence of leveraged (L = 0) outliers
but drastically deteriorates when the outliers are leveraged ($L = 1$). All estimators were more prone to
swapping (higher $\M{FP}$ values) than masking effects. This was also reflected in 
the slightly reduced performance of \textsf{RobRegCC} compared to the standard log-contrast model in the
absence of outliers ($\M{O}=0$). Here, the NR approach achieved the best performance in scaled estimation
error for the regression coefficients ($\M{Er} = 1.04$ for NR compared to $\M{Er} = 1.11$, $\M{Er} = 1.12$,
and $\M{Er} = 1.13$, for penalties $P^A$, $P^H$, and $P^E$, respectively). This is due the fact that the
standard log-contrast model could take all sample into account whereas the other estimators suffered from small
swapping effects.

\begin{figure}[h]
\begin{center}
\includegraphics[width=1\textwidth, angle=0]{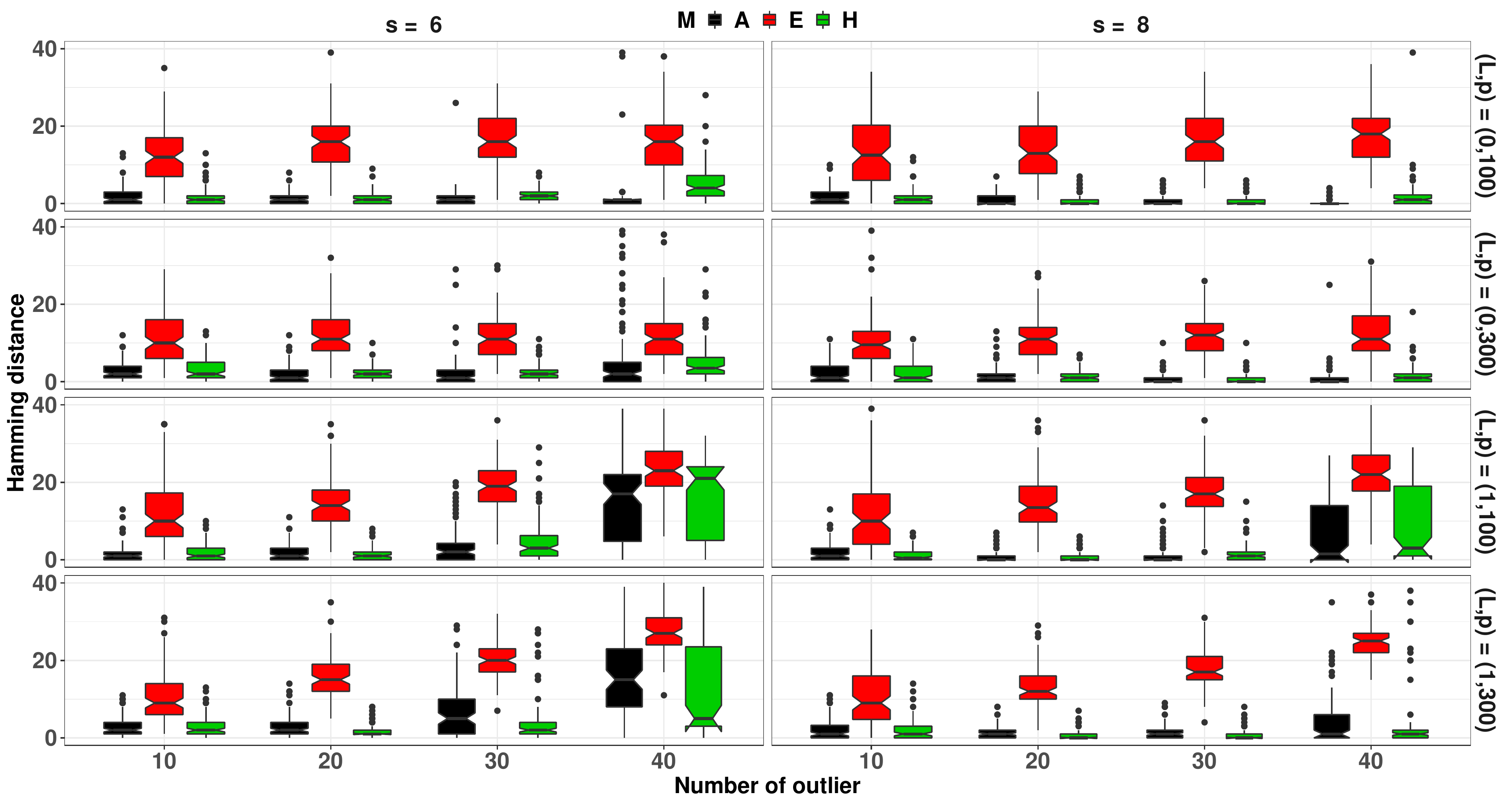}
\caption{\textsf{RobRegCC}'s outlier detection performance in terms of Hamming distance (\textsf{HM}) on simulated data ($p = \{100, 300\}$, $n=200$) with outlier observations shifted by $s = 6\sigma$ (left column) and $s = 8\sigma$ (right column) and number of outliers $\M{O} = \{10,20,30,40\}$. $\M{L} = \{0,1\}$ indicates the absence or presence of leveraged outliers.}
\label{fig:spout}
\end{center}
\end{figure}

While the overall performance of the data-driven R-CV model selection scheme was encouraging, we consistently observed 
slight over-selection of potential outliers in \textsf{RobRegCC} ($F_N \approx 0$ for most settings), in particular with 
the $P^E$ penalty. This behavior was alleviated by the heuristic two-stage estimator whose number of false positives 
$\M{FP}_2$ was consistently lower than $\M{FP}_1$ for \textsf{RobRegCC} with $P^E$. The two-stage estimator thus offers a 
computationally efficient robust estimation alternative when no leveraged outliers are present.  

Figure \ref{fig:spout} summarizes \textsf{RobRegCC}'s overall outlier detection performance across all simulation
scenarios using the Hamming distance. We again observed excellent performance of \textsf{RobRegCC} with hard-ridge(H) and
adaptive penalty(A). The performance decreased only in the setting with a high number of leveraged outliers ($> 
15\%$). The performance of \textsf{RobRegCC} with Elastic Net penalty (E) showed the expected sub-optimal performance
across all scenarios.  

\section{Robust regression on gut microbiome data}
\label{sec:application}
We next applied the \textsf{RobRegCC} workflow to learn robust and predictive models of soluble CD14 
(sCD14) measurements, an immune marker related to chronic inflammation and monocyte activation, from gut microbiome
samples of HIV patients. The data set comprises $n = 151$ observations of sCD14 measurements and aggregated 16S rRNA 
amplicon data across $p = 60$ bacterial genera. In \citet{rivera2018balances}, the data set has been used to highlight the
performance of the balance selection scheme (\texttt{selbal}), a greedy step-wise log-contrast modeling method. 
We provide three comparative analyses on this dataset, showcasing the flexibility of \textsf{RobRegCC}.

\subsection{Comparison of \textsf{RobRegCC} with standard log-contrast approaches}
We modeled the sCD14 measurements as continuous response $\y \in \mathbb{R}^{n}$, considered
the clr transform of the relative genera abundances as compositional covariates $\Z\in \mathbb{R}^{n \times p}$, and
used an intercept $\bN = \1_n$ as non-compositional covariate. 

To facilitate comparison with \texttt{selbal}, we first considered \textsf{RobRegCC} with the standard zero-sum
constraint $\C = \1_p\trans$, analyzed model performance in terms of overall $R^2$, and compared the set of sparse
predictors. For \texttt{selbal}, \citet{rivera2018balances} report a log-contrast model with four genera:
[g]Subdoligranulum and [f]Lachnospiraceae\_[g]\_Incertae\_Sedis (in the numerator) and  
[f]Lachnospiraceae\_[g]\_unclassified and [g]Collinsella (in the denominator). The \texttt{selbal} log-contrast model
fit on all data achieves an $R^2 = 0.281$. \textsf{RobRegCC} identified nine outliers with the $P^A$ and $P^E$
penalties, and five outliers with the $P^H$ penalty, respectively (see Figures \ref{fig:hivmodelfitdiagada1}
--\ref{fig:hivmodelfitdiaghard1} in the Supplementary material). \emph{RobRegCC}
infers slightly less sparse models with eight to ten predictors. After removal of the outliers, \emph{RobRegCC}'s
models achieved considerable higher $R^2$'s, ranging from $0.53$ ($\M{E}$), to $0.57$ ($\M{H}$), and $0.63$
($\M{A}$), respectively. Figure \ref{fig:hivallrobvar} reports the identified set of genera in the respective models. 

\begin{figure}[t]
\begin{center}
\includegraphics[page=1,width=0.9\textwidth]{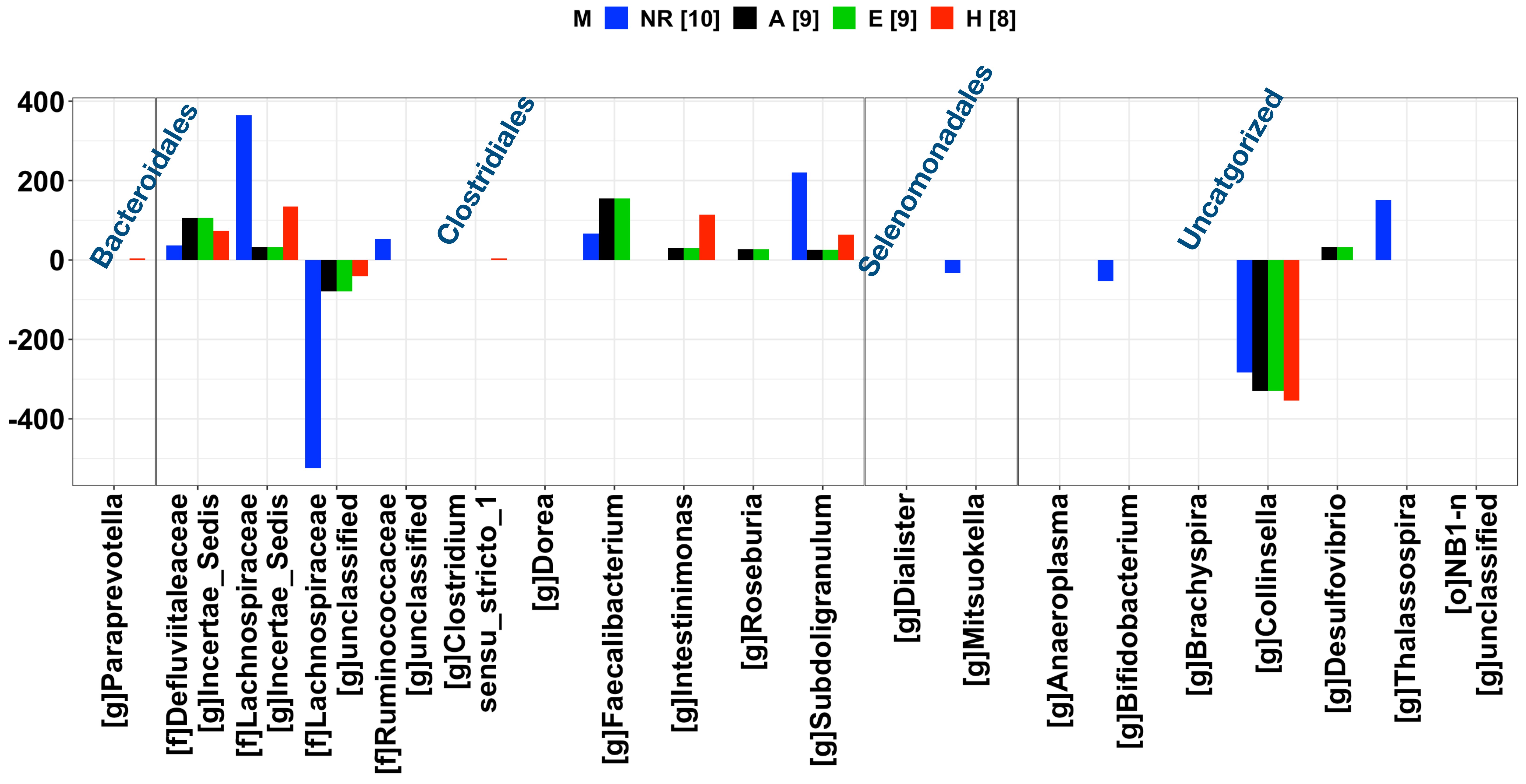}
\caption{\textsf{RobRegCC} log-contrast predictors for sCD14 HIV dataset with zero-sum constraint $\C = \1_p$, the
three penalty functions (H,E,A), and without mean shift (denoted ``non-robust" (NR)) (number of non-zero coefficients
of the different models are shown in parenthesis). Predictor labels correspond to genera names at the finest
taxonomic resolution available (class(c), family(f), genus(g) and order(o)).
}\label{fig:hivallrobvar}
\end{center}
\end{figure}

Consistent with the \texttt{selbal} findings, the robust models included the four genera [g]Subdoligranulum,
[f]Lachnospiraceae\_[g]\_Incertae\_Sedis, [f]Lachnospiraceae\_[g]\_unclassified, and [g]Collinsella.
In addition, all \textsf{RobRegCC} models also identified the genera [f]Defluviitaleacea\_[g]IncertaeSedis and
[g]Intestinimonas to be positively associated with sCD14. The \textsf{RobRegCC} models with the (adaptive) 
Elastic Net penalties identified the genus [g]Faecalibacterium to be positively associated with sCD14 as well.
The standard non-robust [NR] log-contrast model identified several genera, not present in the robust models,
including [g]Bifidobacterium, [g]Mitsukella, and [g]Thalassospira.  

To give a fair evaluation of the out-of-sample predictive performance of the \textsf{RobRegCC} models,
we randomly split the data 100 times into two sets with $90\%$ ($n_\text{tro} = 136$) samples \{$\y_{tro},\,\,
\X_{tro}$\} for training and $10\%$ ($n_\text{teo} = 15$) samples \{$\y_{teo},\,\, \X_{teo}$\} for out-of-sample 
prediction. We used the robust test statistic, introduced in \eqref{eq:teststats}, to measure robust out-of-sample
prediction error. Table \ref{tab:methodcomp} reports the mean and standard deviation (in parenthesis) of the
the robust error $\epsilon_{teo,\lambda_{cv}}$, average sample size $n_{teo,r}$ after outlier removal in the test
data, and the percentage of outliers $\M{O}_{tro}$\% identified in the training phase for both robust and non-robust
models. The comparison also includes the \texttt{selbal} model with the four genera \citep{rivera2018balances} as
predictors, denoted by $\M{NR}_0$. We observed that the robust approaches showed superior estimation performance (i.e.,
lower test error) and identified roughly 5-6\% of the samples as outliers. 
\begin{table}[h]
\caption{Predictive modeling of sCD14 data with $\C = \1_p$: Mean and standard deviation (in parenthesis) of the
out-of-sample scaled test error  $\epsilon_{teo,\lambda_{cv}}$, the ``clean" sample size ($n_{teo,r}$) on the 
test data, and the percentage of outliers $\M{O}_{tro}$\% identified in the training data.}
\label{tab:methodcomp}
\centering
\begin{tabular}{|llllll|}
  \hline
 & A & E & H & NR & $\M{NR}_0$ \\ 
  \hline
$\epsilon_{teo,\lambda_{cv}}$  & 0.28 (0.20) & 0.29 (0.22) & 0.29 (0.20) & 0.36 (0.20) & 0.34 (0.18) \\ 
$n_{teo,r}$ & 13.18 (1.52) & 13.68 (0.98) & 13.22 (1.46) & 13.41 (1.13) & 13.64 (1.01) \\ 
$\M{O}_{tro}$\%  & 4.84 (1.76) & 6.09 (2.05) & 4.78 (1.39) & 0.00 (0.00) & 0.00 (0.00) \\ 
   \hline
\end{tabular}
\end{table}

\subsection{Robust regression with subcompositional coherence}
The process of measuring relative microbial species abundances introduces biases 
at multiple experimental stages \citep{pollock2018madness}, including taxonomy-dependent biases 
due to some microbes being more resistant to cell lysis or variable specificities of the primer sets. 
Such taxonomic biases can be mitigated by enforcing sub-compositional coherence with respect to taxonomic grouping
\citep{shi2016regression}. Here, we extended the analysis from before with the subcompositional coherence 
imposed at the order level, resulting in a constraint matrix $\C$ (see Equation
\ref{eq:defCmathiv} in the Supplementary Material) with $k=6$ subcompositions. The taxa with known order
information were grouped into five subcompositions. Uncategorized taxa formed the sixth subcomposition.

\begin{figure}[t]
\begin{center}
\includegraphics[page=2,width=0.9\textwidth]{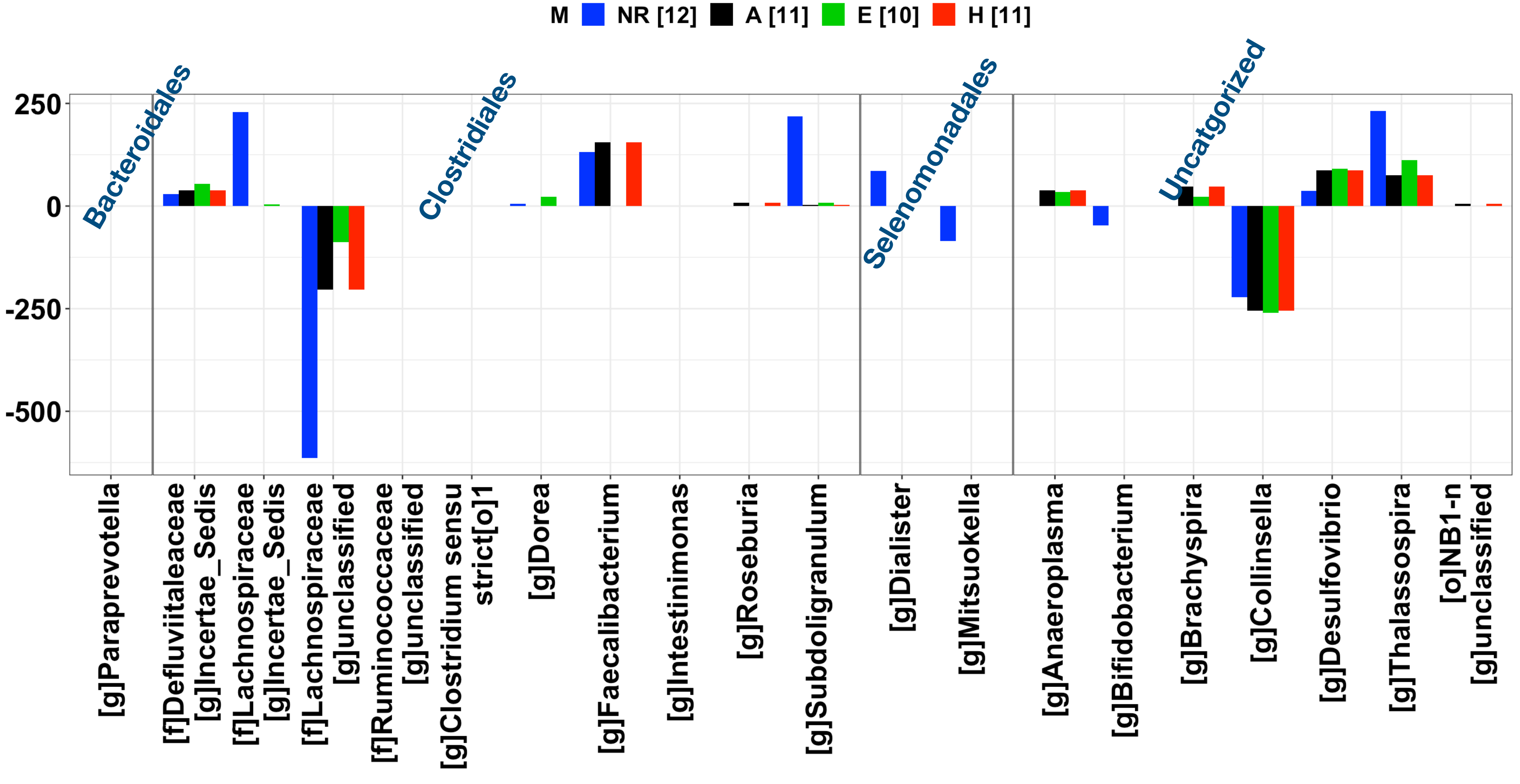}
\caption{\textsf{RobRegCC} log-contrast predictors on the sCD14 HIV dataset with subcompositional constraints on order level, the
three penalty functions (H,E,A), and without mean shift (denoted ``non-robust" (NR)) (number of non-zero coefficients
of the different models are shown in parenthesis). Predictor labels correspond to genera names at the finest
taxonomic resolution available (class(c), family(f), genus(g) and order(o)).
}\label{fig:hivallrobvar_subcomp}
\end{center}
\end{figure}

\textsf{RobRegCC} identified nine outliers with the $P^A$ and the $P^H$ penalties and ten outliers with $P^E$,
respectively. The subcompositional constraint induced slightly denser models with ten to twelve predictors while 
simultaneously maintaining superior out-of-sample prediction performance (see Table \ref{tab:methodcompo}) and model $R^2$'s (see Figure \ref{fig:hivmodelfitdiagada} --\ref{fig:hiv_fit__r3} in the Supplementary Material), when compared to the \texttt{selbal} or the non-robust model. 

\begin{table}[h]
\caption{Predictive modeling of sCD14 data with subcompositional constraint: 
Mean and standard deviation (in parenthesis) of the out-of-sample scaled test error 
$\epsilon_{teo,\lambda_{cv}}$, the ``clean" sample size ($n_{teo,r}$) on the 
test data, and the percentage of outliers $\M{O}_{tro}$\% identified in the training data.}
\label{tab:methodcompo}
\centering
\begin{tabular}{|llllll|}
  \hline
 & A & E & H & NR & $\M{NR}_0$ \\ 
  \hline
$\epsilon_{teo,\lambda_{cv}}$ & 0.29 (0.19) & 0.26 (0.19) & 0.29 (0.20) & 0.34 (0.18) & 0.34 (0.19) \\ 
$n_{teo,r}$ & 13.66 (0.99) & 13.61 (0.96) & 13.61 (0.99) & 13.66 (1.04) & 13.71 (0.96) \\ 
$\M{O}_{tro}$\%  & 4.88 (1.58) & 6.28 (1.91) & 4.78 (1.75) & 0.00 (0.00) & 0.00 (0.00) \\ 
   \hline
\end{tabular}
\end{table}

Figure \ref{fig:hivallrobvar_subcomp} reports the selected microbial species that were associated with the sCD14 inflammation marker. With the subcompositional coherence at order level, [f]Defluviitaleacea\_[g]IncertaeSedis,
[f]Lachnospiraceae\_[g]IncertaeSedis, and [g]Faecalibacterium were associated with sCD14. The models also include the genera [g]Desilfovibrio and [g]Thalassospira and discard the genus [f]Lachnospiraceae\_[g]IncertaeSedis, when compared to the previous analysis. The robust models only selected predictors 
in the Clostridiales and the Uncategorized subcomposition.

\subsection{Robustness to data mislabeling}
A common source of error in analyzing microbial datasets, in particular those coming from public resources 
such as NCBI's SRA (\url{https://www.ncbi.nlm.nih.gov/sra}), stems from insufficient documentation of 
the correspondence between data files comprising raw read data and their associated experimental meta-information.
Missing or mislabeled meta-information is not uncommon and hinders large-scale meta- or re-analysis of many public data.
To show that \textsf{RobRegCC} can deal with potential data mislabeling, we emulated such a scenario on the 
sCD14 dataset by generating $\M{O} = 10$ mislabeled observations. We actively interchanged the $O/2$ largest and $O/2$ 
smallest entries in the response $\y$ while keeping the corresponding rows in $\Z$ unchanged (see Figure 
\ref{fig:mislabeling}(a)). We observed that non-robust regression on the mislabeled data resulted in a significant drop in 
model fit ($R^2 = 0.065$, see Figure \ref{fig:hiv_fit__r2c} in the Supplementary Material). \textsf{RobRegCC}'s 
performance was not affected by the corrupted observations ($R^2 = 0.69$ ($\M{H}$), $0.55$ ($\M{E}$), and $0.55$
($\M{A}$), see Figure \ref{fig:hiv_fit__r2c} in the Supplementary Material). We further compared the similarity 
between \textsf{RobRegCC}'s predictors on the original data $\what{\bbeta}_o$ and the predictors $\what{\bbeta}_m$ on 
the mislabeled data by measuring the relative error $\M{err}_{rel} = 100\|\what{\bbeta}_o -
\what{\bbeta}_m\|/\|\what{\bbeta}_o\|$ and the support mismatch via the Hamming distance $\M{HM}{(\hat \bbeta_m})$. Figures
\ref{fig:mislabeling}(b),(c) summarize the error estimates. The \textsf{RobRegCC} model with adaptive Elastic Net penalty outperformed all other methods in both error measures, shared five predictors with the
regression model on the original data, and correctly identified seven mislabeled data points as outliers (Figure \ref{fig:mislabeling}(d)). 

\begin{figure}[t]
\centering
\subfloat[Mislabeling experiment]{\includegraphics[page = 1, width=0.32\textwidth]{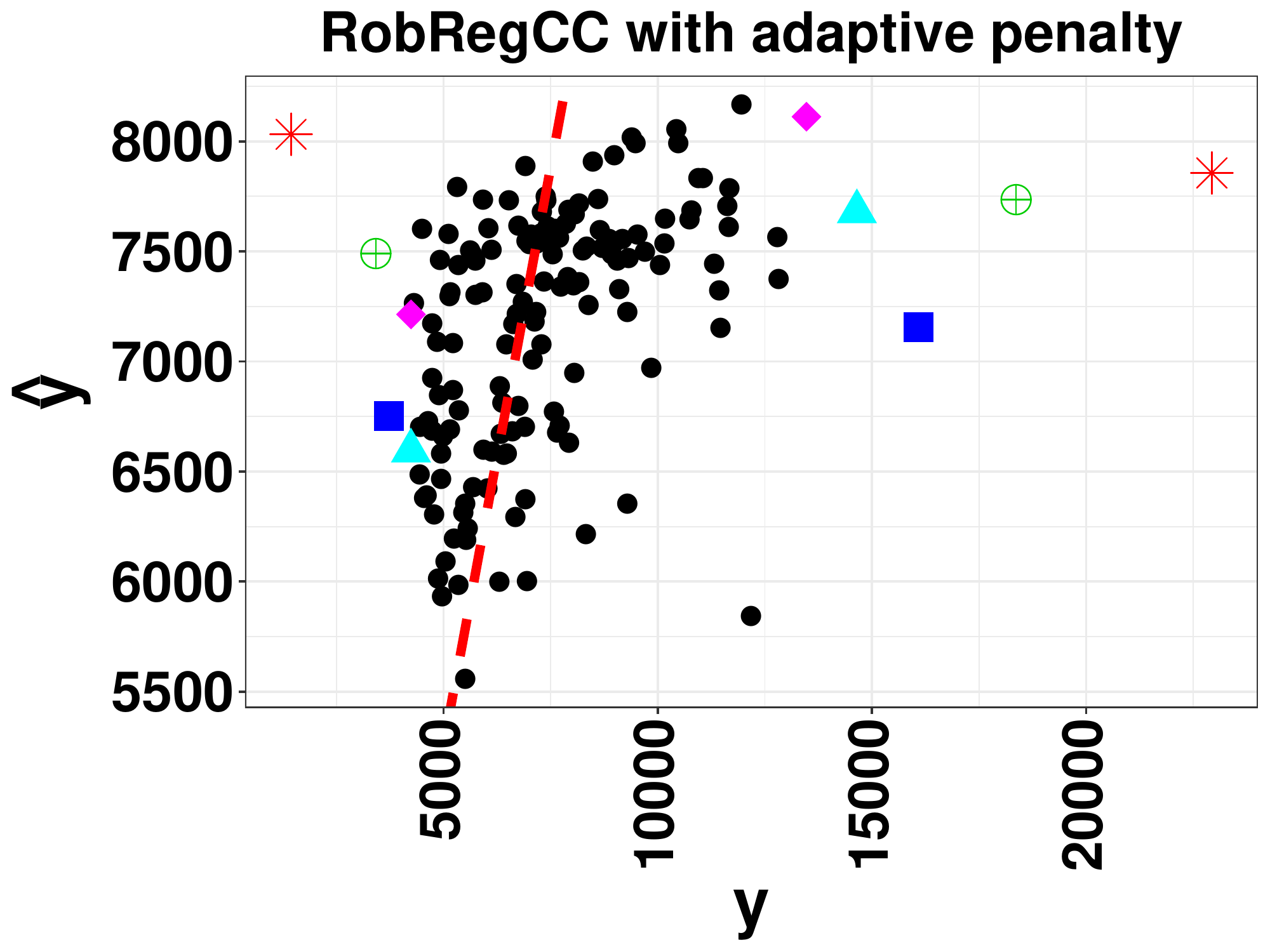}} 
\hfill
\subfloat[$\M{err}_{rel}$]{\includegraphics[page = 1, width=0.16\textwidth]{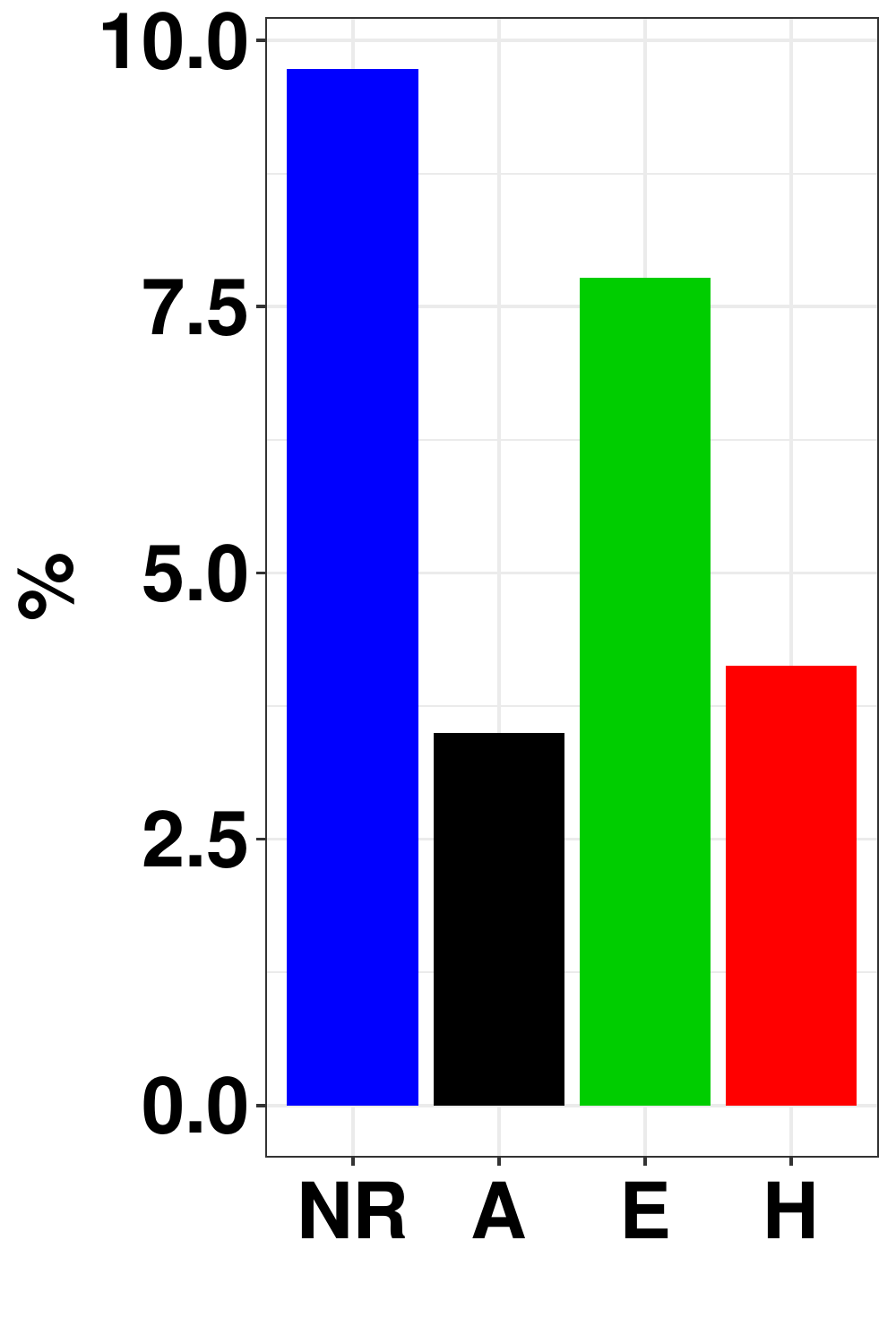}}
\hfill
\subfloat[$\M{HM}{(\hat \bbeta_m}) $]{\includegraphics[page = 2, width=0.16\textwidth]{appl_hiv_delbeta.pdf}}
\hfill
\subfloat[Residual]{\includegraphics[page = 1, width=0.32\textwidth]{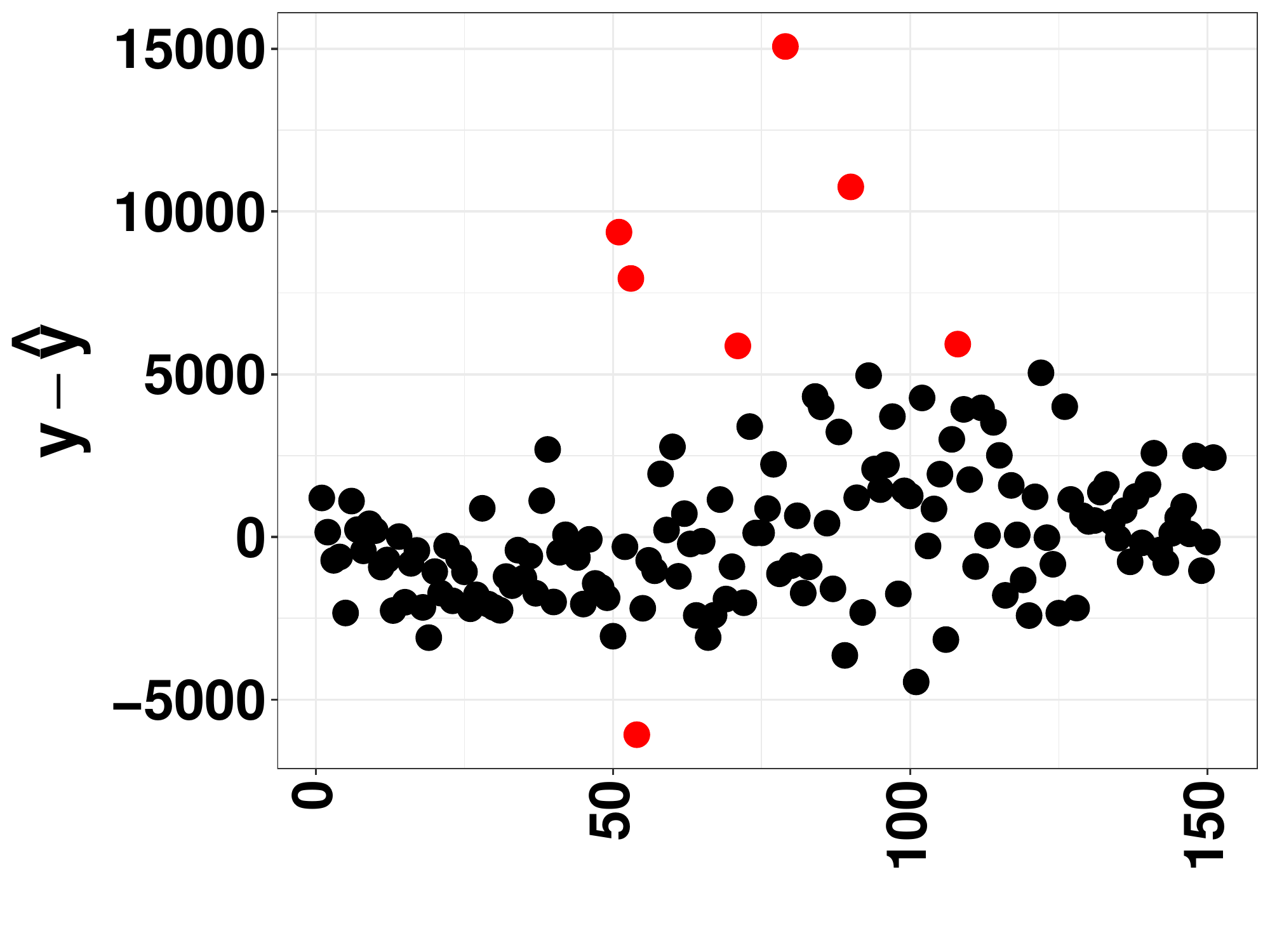}
}
\caption{Comparison of the robust and  non-robust approaches on mislabeled data:
(a) Color coded observations of the sCD14 measurements ($\y$) on the two side of the red dotted line are actively interchanged; (b) relative error $\M{err}_{rel} = 100\|\what{\bbeta}_o - \what{\bbeta}_m\|/\|\what{\bbeta}_o\|$ (see main text for explanation);(c) support mismatch $\M{HM}{(\bbeta_m})$ (see main text for explanation);(d) estimated residual plot with outliers identified in red dots using \textsf{RobRegCC} with adaptive Elastic Net penalty.}
\label{fig:mislabeling}
\end{figure}

\section{Discussion and conclusion}
In this contribution, we have presented \textsf{RobRegCC}, a robust log-contrast regression framework that
allows simultaneous outlier and sparse model coefficient identification for regression problems with
compositional and non-compositional covariates. The approach combines the idea of mean shift 
estimation in linear regression with robust initialization and penalization for linear log-contrast 
regression \citep{aitchison1984log}. We have tackled the resulting over-specified model parameter
estimation problem via regularization with suitable sparsity-inducing penalty functions,  
including the hard-ridge, the Elastic Net, and a novel adaptive Elastic Net penalty. 
While the estimation approach with the Elastic Net penalty lacks the ability to handle masking and 
swapping effect \citep{she2011outlier}, the adaptive Elastic Net and the hard-ridge penalties 
alleviate this problem but require initial robust estimates of the parameters 
to construct appropriate Elastic Nets weights or good initial parameter estimates, respectively.  
For the robust initialization step, we have used the concept of principal 
sensitivity component analysis \citep{pena1999fast}. \textsf{RobRegCC} also includes (i) a general 
Lagrangian-based optimization procedure to solve the underlying optimization
problem with any of the available penalty functions and (ii) a novel robust prediction error measure and
cross-validation scheme that may be of independent interest.
We have shown, on simulated and real compositional microbiome data, the validity and generality of our approach 
and have developed novel theoretical results that give prediction error bounds for the \textsf{RobRegCC}
estimators in the finite sample setting. In practice, we recommend using \textsf{RobRegCC} with the adaptive
Elastic Net penalty, since this estimator showed superior prediction and consistency performance on 
the majority of the experimental scenarios. 

Future computational efforts will include exploring and implementing other computationally 
efficient optimization strategies, including recent path-based algorithms for log-contrast regression
\citep{Gaines2018}. On the theoretical side, we will analyze the variable selection properties of the
RobRegCC model estimators. A natural extension of our modeling framework is robust logistic regression
when responses are given as class indicators rather than continuous variables. In summary, we believe that 
our \textsf{RobRegCC} framework provides a useful tool for statisticians and computational biologists 
that want to robustly solve regression problems with compositional covariates.


\bibliographystyle{biorefs}
\bibliography{library}

\section{Supplementary Material}
\label{sec:supp}
\setcounter{section}{1}
\section*{Supplementary material}\label{sec:appen}
\beginsupplement

\subsection{Simulation: Model performance in the simulation setting with moderate outliers}

\begin{table}[H]

\caption{\label{tab:sequal6}Comparison of the non-robust [NR] model and the robust model, i.e.,  RobRegCC with the hard-ridge [H], the Elastic Net [E], and the adaptive Elastic Net [A] penalty function, in the simulation setting with moderate outliers (n = 200, s = 6) using the outlier identification measures false negative (FN) and false positive (FP), and the estimation error measure Er($\bbeta$). Here $\M{FP}_1$ and $\M{FP}_2$ refers to the pre and post false positive measures.\\}
\centering
\resizebox{\linewidth}{!}{
\fontsize{7.5}{9.5}\selectfont
\begin{tabu} to \linewidth {>{\raggedleft}X>{\raggedleft}X>{\raggedleft}X>{\raggedright}X>{\raggedright}X>{\raggedright}X>{\bfseries}l>{\raggedright}X>{\raggedright}X>{\raggedright}X>{\bfseries}l>{\raggedright}X>{\raggedright}X>{\raggedright}X>{\bfseries}l>{\bfseries}l}
\toprule
\multicolumn{3}{c}{\textbf{ }} & \multicolumn{4}{c}{\textbf{[A]}} & \multicolumn{4}{c}{\textbf{[H]}} & \multicolumn{4}{c}{\textbf{[E]}} & \multicolumn{1}{c}{\textbf{[NR]}} \\
\cmidrule(l{3pt}r{3pt}){4-7} \cmidrule(l{3pt}r{3pt}){8-11} \cmidrule(l{3pt}r{3pt}){12-15} \cmidrule(l{3pt}r{3pt}){16-16}
\multicolumn{3}{c}{ } & \multicolumn{3}{c}{$\bgamma$} & \multicolumn{1}{c}{$\bbeta$} & \multicolumn{3}{c}{$\bgamma$} & \multicolumn{1}{c}{$\bbeta$} & \multicolumn{3}{c}{$\bgamma$} & \multicolumn{1}{c}{$\bbeta$} & \multicolumn{1}{c}{$\bbeta$} \\
\cmidrule(l{3pt}r{3pt}){4-6} \cmidrule(l{3pt}r{3pt}){7-7} \cmidrule(l{3pt}r{3pt}){8-10} \cmidrule(l{3pt}r{3pt}){11-11} \cmidrule(l{3pt}r{3pt}){12-14} \cmidrule(l{3pt}r{3pt}){15-15} \cmidrule(l{3pt}r{3pt}){16-16}
\rowcolor{gray!6}  L & p & O & FN & $\M{FP}_1$ & $\M{FP}_2$ & Er & FN & $\M{FP}_1$ & $\M{FP}_2$ & Er & FN & $\M{FP}_1$ & $\M{FP}_2$ & Er & Er\\
\midrule
0 & 100 & 10 & 0.00 & 1.78 & 1.13 & 1.11 & 0.30 & 0.74 & 0.62 & 1.09 & 0.00 & 12.13 & 3.07 & 1.17 & 1.51\\
\rowcolor{gray!6}  0 & 100 & 20 & 0.00 & 1.10 & 0.67 & 1.14 & 0.46 & 0.30 & 0.47 & 1.13 & 0.00 & 15.54 & 3.48 & 1.19 & 1.91\\
0 & 100 & 30 & 0.00 & 1.00 & 0.57 & 1.07 & 1.63 & 0.00 & 0.32 & 1.11 & 0.00 & 16.92 & 4.12 & 1.09 & 2.29\\
\rowcolor{gray!6}  0 & 100 & 40 & 0.36 & 0.38 & 0.33 & 1.16 & 4.33 & 0.00 & 0.00 & 1.32 & 0.00 & 14.80 & 2.69 & 1.19 & 2.59\\
\hline
0 & 300 & 10 & 0.00 & 2.53 & 1.18 & 0.51 & 0.00 & 2.72 & 1.18 & 0.52 & 0.00 & 11.15 & 2.68 & 0.54 & 0.81\\
\addlinespace
\rowcolor{gray!6}  0 & 300 & 20 & 0.00 & 1.41 & 0.93 & 0.55 & 0.41 & 1.10 & 0.59 & 0.54 & 0.00 & 11.73 & 2.71 & 0.56 & 1.03\\
0 & 300 & 30 & 0.28 & 0.98 & 0.41 & 0.62 & 1.68 & 0.35 & 0.32 & 0.59 & 0.00 & 10.57 & 1.91 & 0.59 & 1.19\\
\rowcolor{gray!6}  0 & 300 & 40 & 1.29 & 0.25 & 0.27 & 0.73 & 3.61 & 0.00 & 0.00 & 0.69 & 0.17 & 9.88 & 1.83 & 0.65 & 1.19\\
\hline
1 & 100 & 10 & 0.00 & 1.17 & 0.96 & 1.10 & 0.00 & 0.90 & 0.62 & 1.10 & 0.91 & 11.35 & 1.97 & 1.21 & 1.77\\
\rowcolor{gray!6}  1 & 100 & 20 & 0.26 & 0.45 & 0.49 & 1.00 & 0.99 & 0.24 & 0.37 & 1.02 & 3.81 & 10.31 & 2.54 & 1.74 & 1.93\\
\addlinespace
1 & 100 & 30 & 1.23 & 0.43 & 0.47 & 1.31 & 3.66 & 0.00 & 0.00 & 1.29 & 8.89 & 10.13 & 2.37 & 1.80 & 2.12\\
\rowcolor{gray!6}  1 & 100 & 40 & 14.35 & 0.47 & 0.55 & 1.73 & 14.92 & 0.36 & 0.43 & 1.82 & 13.48 & 9.87 & 1.96 & 1.90 & 2.49\\
\hline
1 & 300 & 10 & 0.00 & 2.36 & 0.97 & 0.52 & 0.00 & 2.29 & 1.01 & 0.52 & 1.77 & 7.94 & 2.41 & 0.59 & 0.87\\
\rowcolor{gray!6}  1 & 300 & 20 & 0.00 & 1.14 & 0.54 & 0.52 & 0.58 & 0.40 & 0.39 & 0.49 & 6.28 & 8.99 & 2.44 & 0.96 & 0.93\\
1 & 300 & 30 & 5.35 & 0.36 & 0.36 & 0.73 & 1.85 & 0.00 & 0.42 & 0.61 & 12.82 & 6.83 & 1.89 & 0.96 & 0.91\\
\addlinespace
\rowcolor{gray!6}  1 & 300 & 40 & 14.66 & 0.00 & 0.36 & 0.86 & 11.92 & 0.00 & 0.00 & 0.77 & 18.51 & 6.93 & 2.18 & 0.94 & 0.94\\
\bottomrule
\end{tabu}}
\end{table}

%
%
%
%
%

\newpage
\subsection{HIV data analysis}

\subsubsection{Robust HIV data analysis with $\C = \1_p$}

Figure \ref{fig:hivmodelfitdiagada1} -- \ref{fig:hivmodelfitdiaghard1} shows the model fit diagnostic with \textsf{RobRegCC} in analyzing the HIV microbial abundance data to explore its  association with the immune  inflammation marker CD14. 

\begin{figure}[H]
\begin{center}
\includegraphics[page=1, width=0.8\textwidth, angle=0]{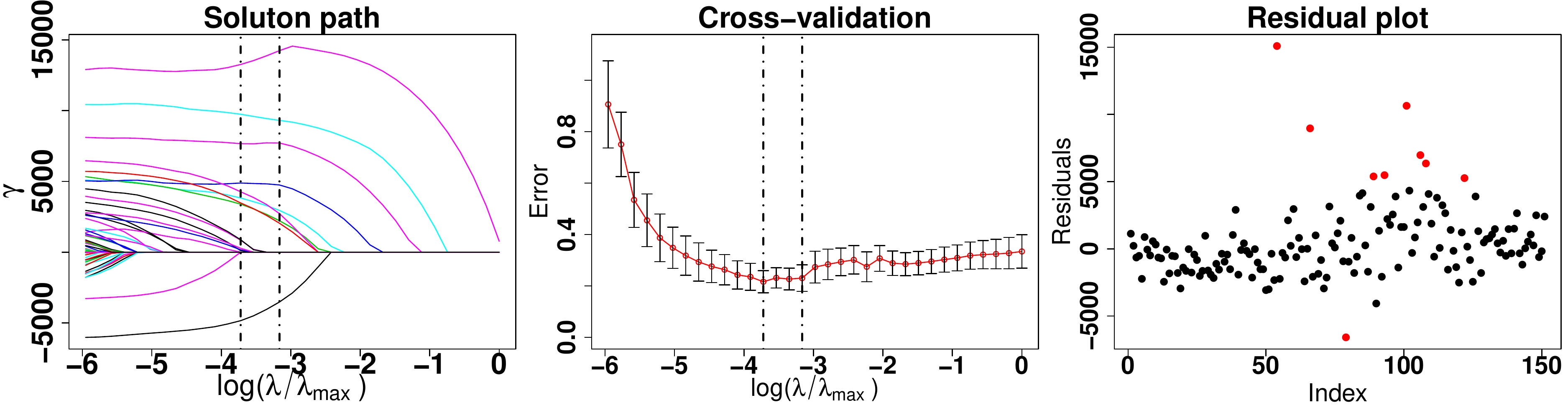}
\caption{\textsf{Robust HIV data analysis with adaptive elastic net penalty}: For $\C = \1_p$,  left and middle plot shows the solution path, representing mean shift parameter  $\bgamma$ estimate, and cross-validation error with varying tuning parameter $\lambda$ using \textbf{RobRegCC} model with adaptive Elastic-net penalty, respectively. Vertical dashed lines are corresponding to the minimum  and one standard error rule test error.  Indentified outliers (red) are depicted using the residual plot (right).}\label{fig:hivmodelfitdiagada1}
\end{center}
\end{figure}

\begin{figure}[H]
\begin{center}
\includegraphics[page=2, width=0.8\textwidth, angle=0]{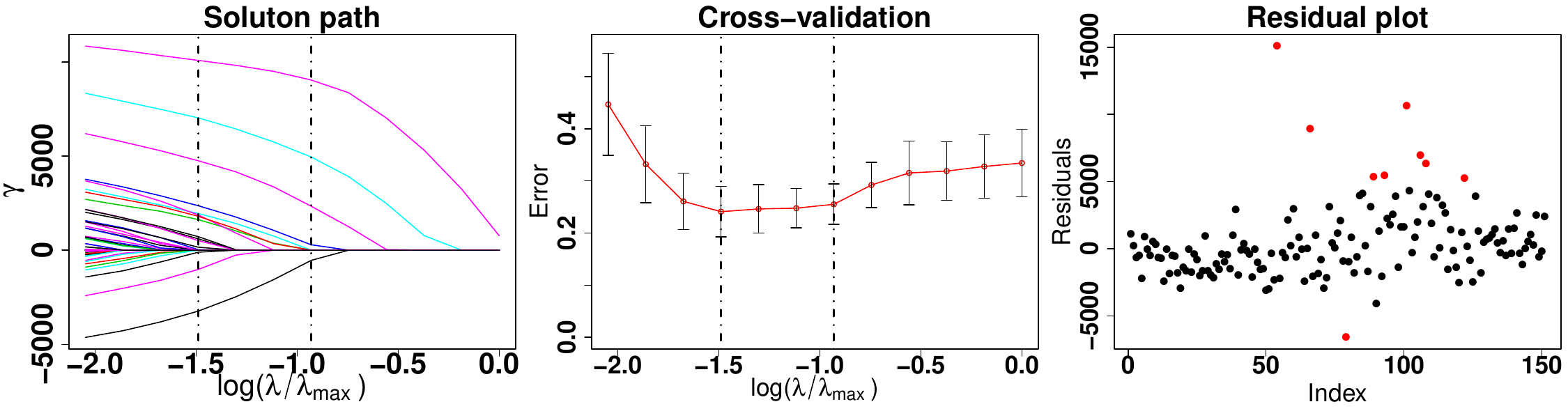}
\caption{\textsf{Robust HIV data analysis with elastic net penalty}: For $\C = \1_p$,  left  and middle plot show solution path, representing mean shift parameter  $\bgamma$ estimate, and cross-validation error with varying tuning parameter $\lambda$ using \textbf{RobRegCC} model with Elastic-net penalty, respectively. Vertical dashed lines are corresponding to minimum  and one standard error rule test error.  Indentified outliers (red) are depicted using the residual plot (right).}\label{fig:hivmodelfitdiagsoft1}
\end{center}
\end{figure}

\begin{figure}[H]
\begin{center}
\includegraphics[page=3, width=0.8\textwidth, angle=0]{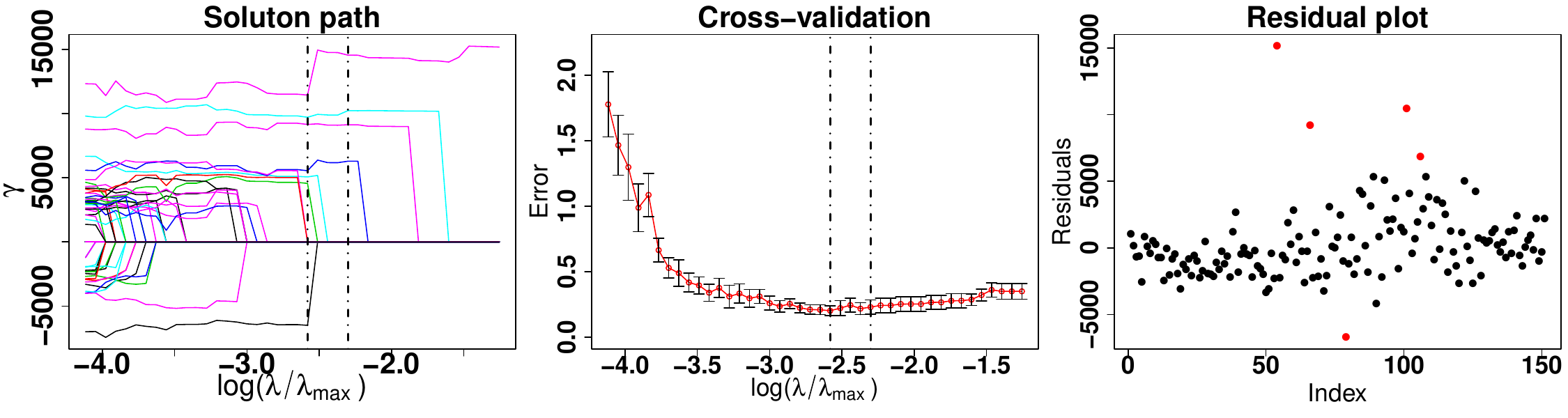}
\caption{\textsf{Robust HIV data analysis with hard ridge penalty}: For $\C = \1_p$,  left  and middle plot show solution path, representing mean shift parameter  $\bgamma$ estimate, and cross-validation error with varying tuning parameter $\lambda$ using \textbf{RobRegCC} model with hard ridge penalty, respectively. Vertical dashed lines are corresponding to minimum  and one standard error rule test error.  Indentified outliers (red) are depicted using the residual plot (right).}\label{fig:hivmodelfitdiaghard1}
\end{center}
\end{figure}

\begin{figure}[H]
\centering
\subfloat[A]{%
\includegraphics[page = 1, width=0.32\textwidth]{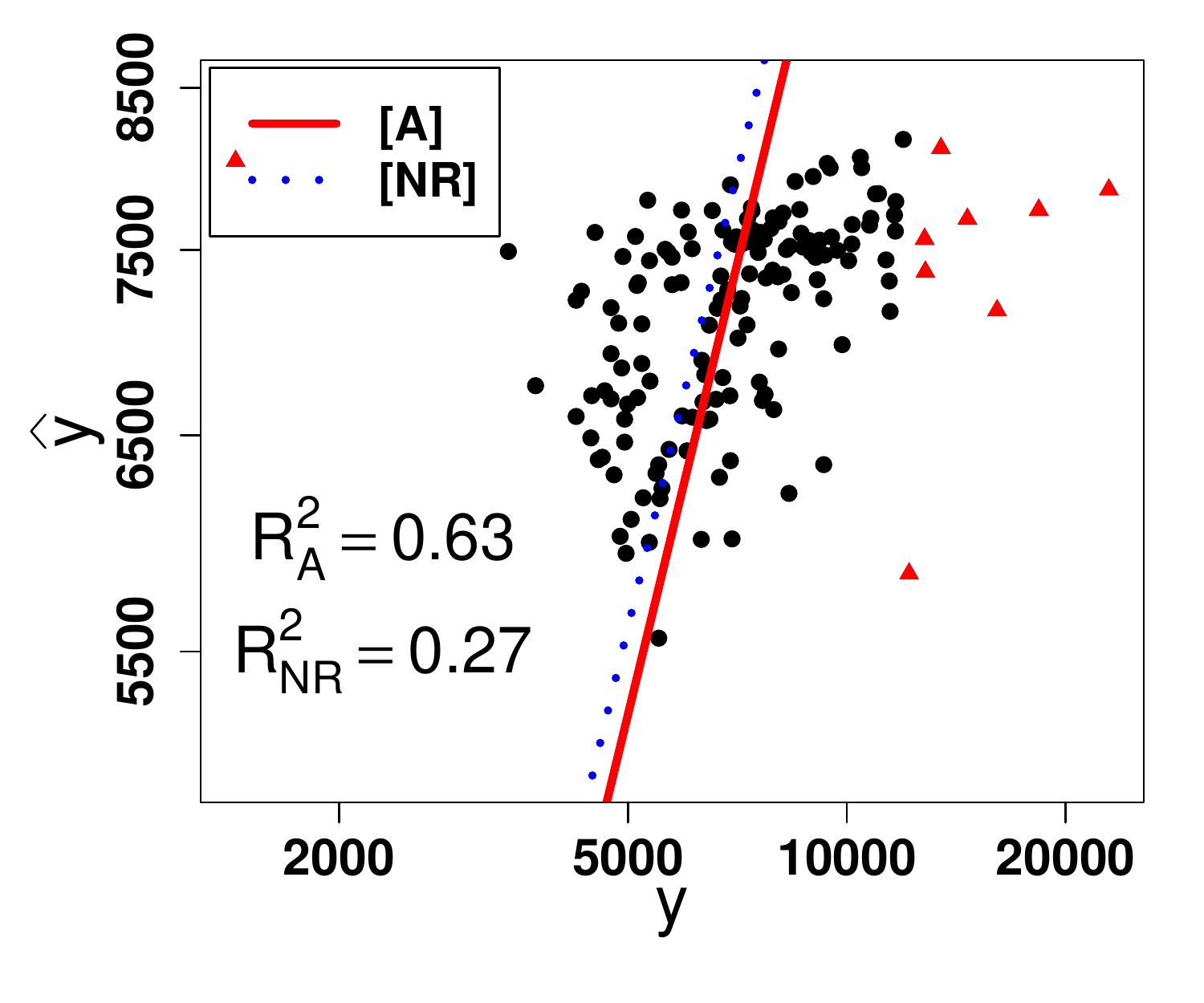}}
\hfill
\subfloat[E]{%
\includegraphics[page = 2, width=0.32\textwidth]{hiv_selbal_fit.pdf}}
\hfill
\subfloat[H]{%
\includegraphics[page = 3, width=0.32\textwidth]{hiv_selbal_fit.pdf}}
     \caption{Application - Robust HIV data analysis  for $\C = \1_p$: Comparison of the robust and  non-robust approach in terms of the model fit statistics $R^2$ and outliers identified (red). In the robust procedure using RobRegCC, $R^2 = 1 - \|\bepsilon_{\lambda_{cv}}\|^2/ \|\y - \bar{y}\|^2$ where   $\bepsilon_{\lambda_{cv}} = \y - \X\what{\btheta}_{\lambda_{cv}} - \what{\bgamma}_{\lambda_{cv}}$ and $\bar{y}$ mean of $\y$.}
     \label{fig:hiv_fit__r2}
   \end{figure}

\subsubsection{Robust HIV data analysis with $\C = \1_p$ after corrupting the responses $\y$}
We corrupt $O = 10$ observations (see main manuscript for the procedure) in the response $\y$, denoting soluble CD14 marker. 
Figure \ref{fig:hivmodelfitdiagada1c} -- \ref{fig:hivmodelfitdiaghard1c} shows the model fit diagnostic with \textsf{RobRegCC} in analyzing the HIV microbial abundance data to explore its  association with the immune  inflammation marker CD14. 

\begin{figure}[H]
\begin{center}
\includegraphics[page=1, width=0.8\textwidth, angle=0]{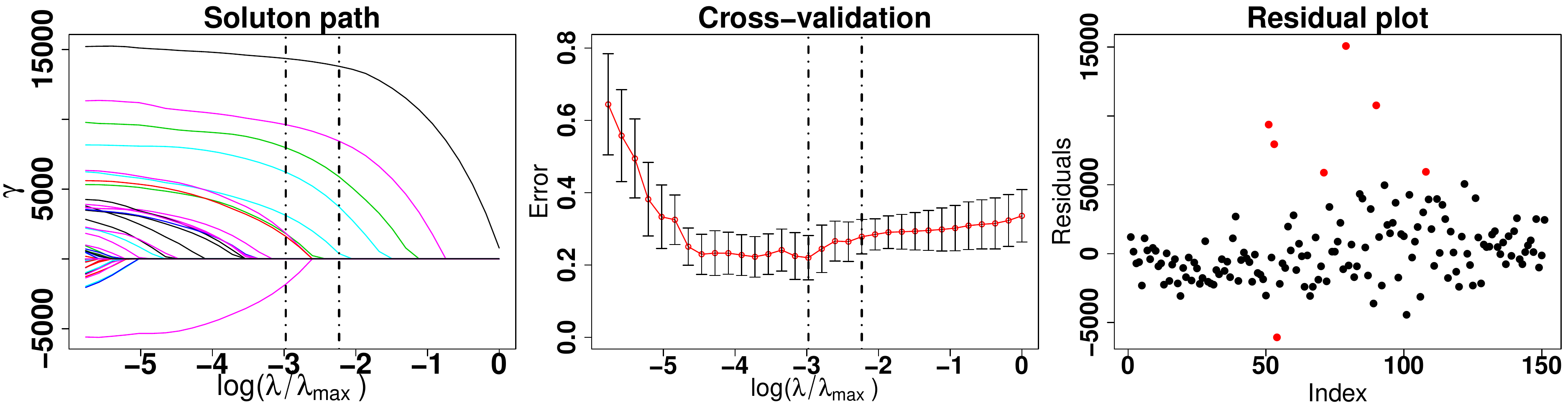}
\caption{\textsf{Robust HIV data analysis with adaptive elastic net penalty with corrupted response $\y$}: For $\C = \1_p$,  left and middle plot shows the solution path, representing mean shift parameter  $\bgamma$ estimate, and cross-validation error with varying tuning parameter $\lambda$ using \textbf{RobRegCC} model with adaptive Elastic-net penalty, respectively. Vertical dashed lines are corresponding to the minimum  and one standard error rule test error.  Indentified outliers (red) are depicted using the residual plot (right).}\label{fig:hivmodelfitdiagada1c}
\end{center}
\end{figure}

\begin{figure}[H]
\begin{center}
\includegraphics[page=2, width=0.8\textwidth, angle=0]{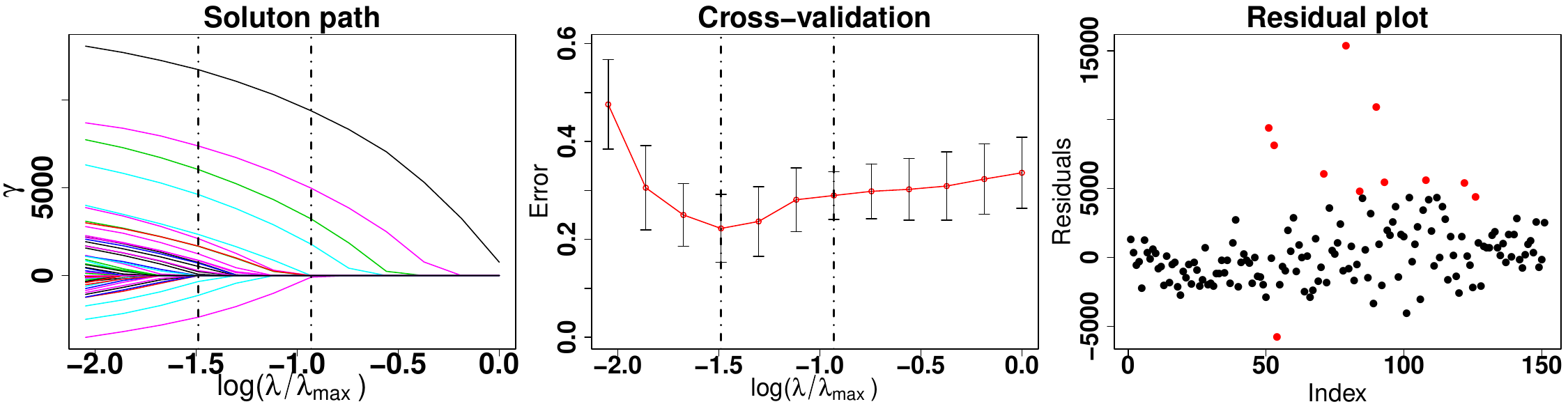}
\caption{\textsf{Robust HIV data analysis with elastic net penalty  with corrupted response $\y$}: For $\C = \1_p$,  left  and middle plot show solution path, representing mean shift parameter  $\bgamma$ estimate, and cross-validation error with varying tuning parameter $\lambda$ using \textbf{RobRegCC} model with Elastic-net penalty, respectively. Vertical dashed lines are corresponding to minimum  and one standard error rule test error.  Indentified outliers (red) are depicted using the residual plot (right).}\label{fig:hivmodelfitdiagsoft1c}
\end{center}
\end{figure}

\begin{figure}[H]
\begin{center}
\includegraphics[page=3, width=0.8\textwidth, angle=0]{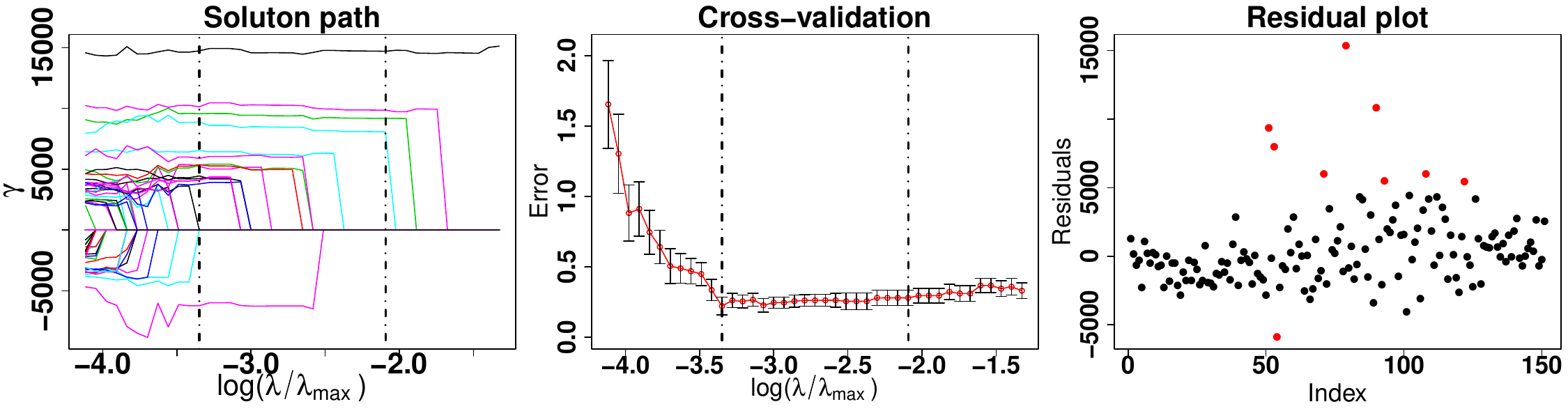}
\caption{\textsf{Robust HIV data analysis with hard ridge penalty  with corrupted response $\y$}: For $\C = \1_p$,  left  and middle plot show solution path, representing mean shift parameter  $\bgamma$ estimate, and cross-validation error with varying tuning parameter $\lambda$ using \textbf{RobRegCC} model with hard ridge penalty, respectively. Vertical dashed lines are corresponding to minimum  and one standard error rule test error.  Indentified outliers (red) are depicted using the residual plot (right).}\label{fig:hivmodelfitdiaghard1c}
\end{center}
\end{figure}

\begin{figure}[H]
\centering
\subfloat[A]{%
\includegraphics[page = 1, width=0.32\textwidth]{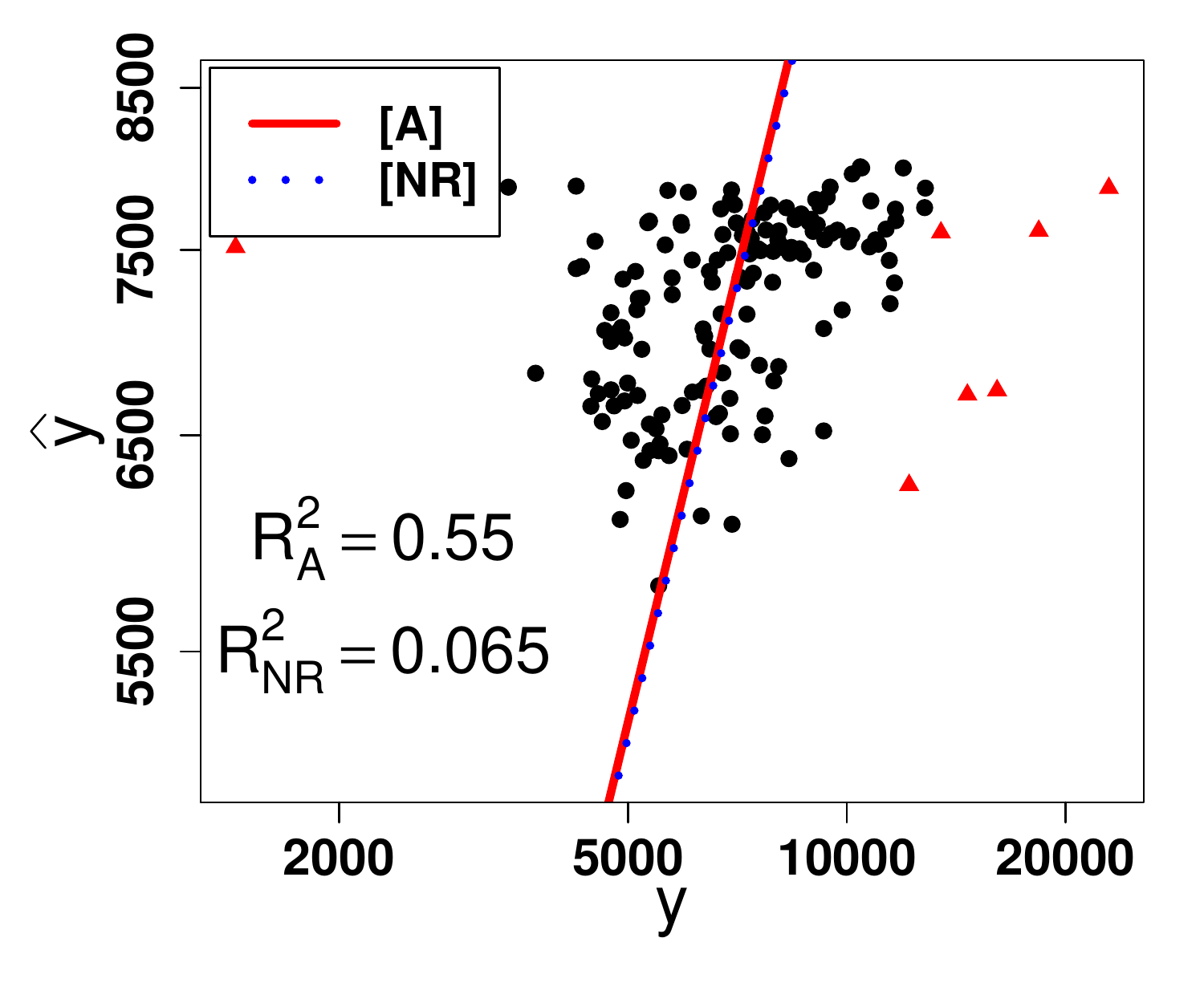}}
\hfill
\subfloat[E]{%
\includegraphics[page = 2, width=0.32\textwidth]{hiv_selbal_fit_c.pdf}}
\hfill
\subfloat[H]{%
\includegraphics[page = 3, width=0.32\textwidth]{hiv_selbal_fit_c.pdf}}
     \caption{Application - Robust HIV data analysis  for $\C = \1_p$: Comparison of the robust and  non-robust approach in terms of the model fit statistics $R^2$ and outliers identified (red). In the robust procedure using RobRegCC, $R^2 = 1 - \|\bepsilon_{\lambda_{cv}}\|^2/ \|\y - \bar{y}\|^2$ where   $\bepsilon_{\lambda_{cv}} = \y - \X\what{\btheta}_{\lambda_{cv}} - \what{\bgamma}_{\lambda_{cv}}$ and $\bar{y}$ mean of $\y$.}
     \label{fig:hiv_fit__r2c}
   \end{figure}

\subsubsection{Robust HIV data analysis with the phylum level subcomposition $\C$}

Figure \ref{fig:hivmodelfitdiagada} -- \ref{fig:hivmodelfitdiaghard} shows the model fit diagnostic with \textsf{RobRegCC} in analyzing the HIV microbial abundance data to explore its  association with the immune inflammation marker CD14. 

The subcomposition matrix for the robust analysis: 
\begin{align}
\C\trans = \begin{bmatrix}
    \1_{p_1}\trans & \0 & \dots  & \0 \\
    \0  & \1_{p_2}\trans &  \dots  & \0 \\
    \vdots & \vdots &  \ddots & \vdots \\
    \0 & \0 &  \dots  & \1_{p_4}\trans 
\end{bmatrix}_{6 \times 60} 
\label{eq:defCmathiv}
\end{align}
with $\s = [0,13, 35,39,45,47,60]$, $p_i = \s_{i+1}-\s_{i}, i=1,\ldots,8$, 
and index sets $\mathbb{A}_{i} = \{\s_{i}+1,\ldots,\s_{i+1}\}$.

\begin{figure}[H]
\begin{center}
\includegraphics[page=1,width=0.8\textwidth, angle=0]{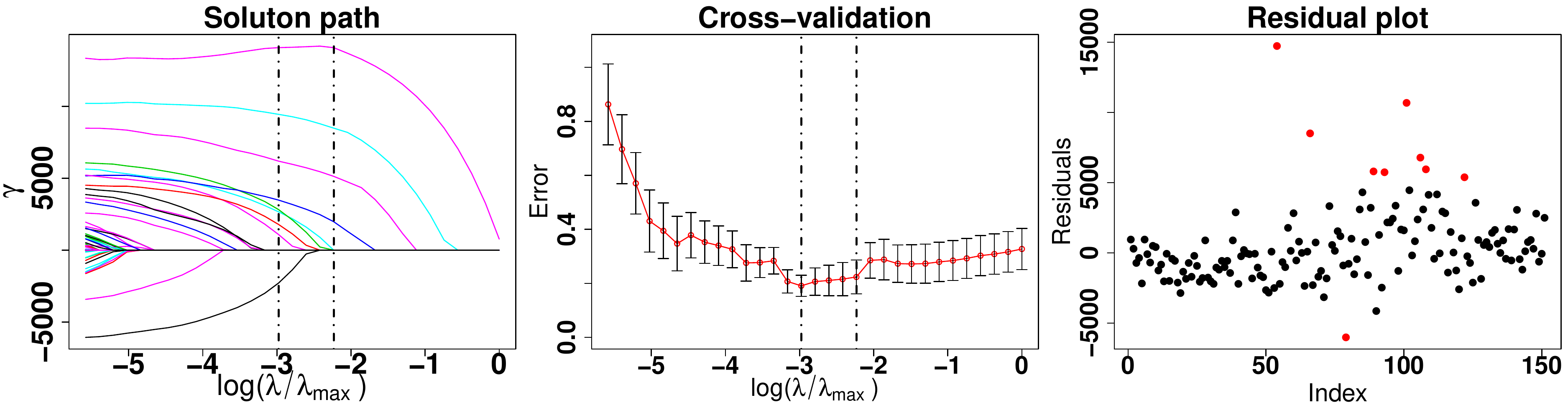}
\caption{\textsf{Robust HIV data analysis with adaptive elastic net penalty}: For the subcomposition matrix $\C$ due to order in the taxonomy,  left and middle plot shows the solution path, representing mean shift parameter  $\bgamma$ estimate, and cross-validation error with varying tuning parameter $\lambda$ using \textbf{RobRegCC} model with adaptive Elastic-net penalty, respectively. Vertical dashed lines are corresponding to the minimum  and one standard error rule test error.  Indentified outliers (red) are depicted using the residual plot (right).}\label{fig:hivmodelfitdiagada}
\end{center}
\end{figure}

\begin{figure}[H]
\begin{center}
\includegraphics[page=2,width=0.8\textwidth, angle=0]{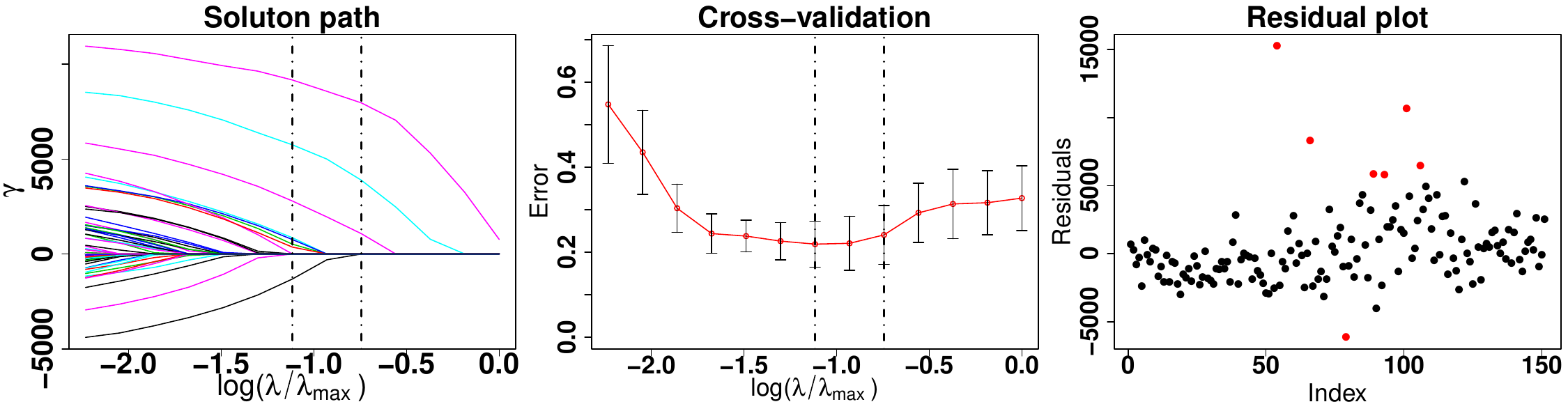}
\caption{\textsf{Robust HIV data analysis with elastic net penalty}:  For the subcomposition matrix $\C $ due to order in the taxonomy,  left  and middle plot show solution path, representing mean shift parameter  $\bgamma$ estimate, and cross-validation error with varying tuning parameter $\lambda$ using \textbf{RobRegCC} model with Elastic-net penalty, respectively. Vertical dashed lines are corresponding to minimum  and one standard error rule test error.  Indentified outliers (red) are depicted using the residual plot (right).}\label{fig:hivmodelfitdiagsoft}
\end{center}
\end{figure}

\begin{figure}[H]
\begin{center}
\includegraphics[page=3,width=0.8\textwidth, angle=0]{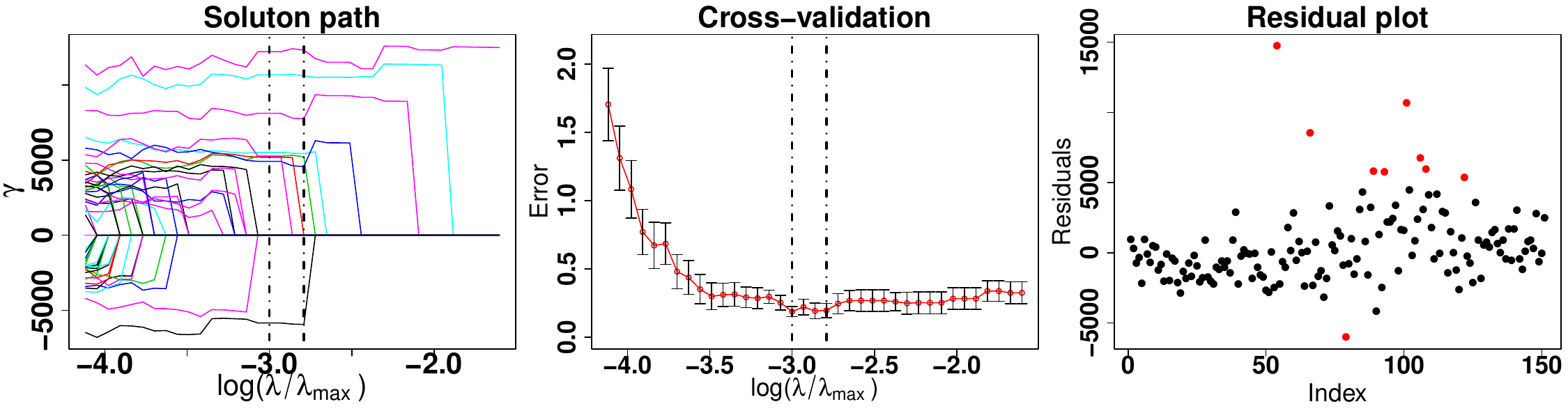}
\caption{\textsf{Robust HIV data analysis with hard ridge penalty}:  For the subcomposition matrix $\C $ due to order in the taxonomy  left  and middle plot show solution path, representing mean shift parameter  $\bgamma$ estimate, and cross-validation error with varying tuning parameter $\lambda$ using \textbf{RobRegCC} model with hard ridge penalty, respectively. Vertical dashed lines are corresponding to minimum  and one standard error rule test error.  Indentified outliers (red) are depicted using the residual plot (right).}\label{fig:hivmodelfitdiaghard}
\end{center}
\end{figure}

\begin{figure}[H]
\centering
\subfloat[A]{%
\includegraphics[page = 1, width=0.32\textwidth]{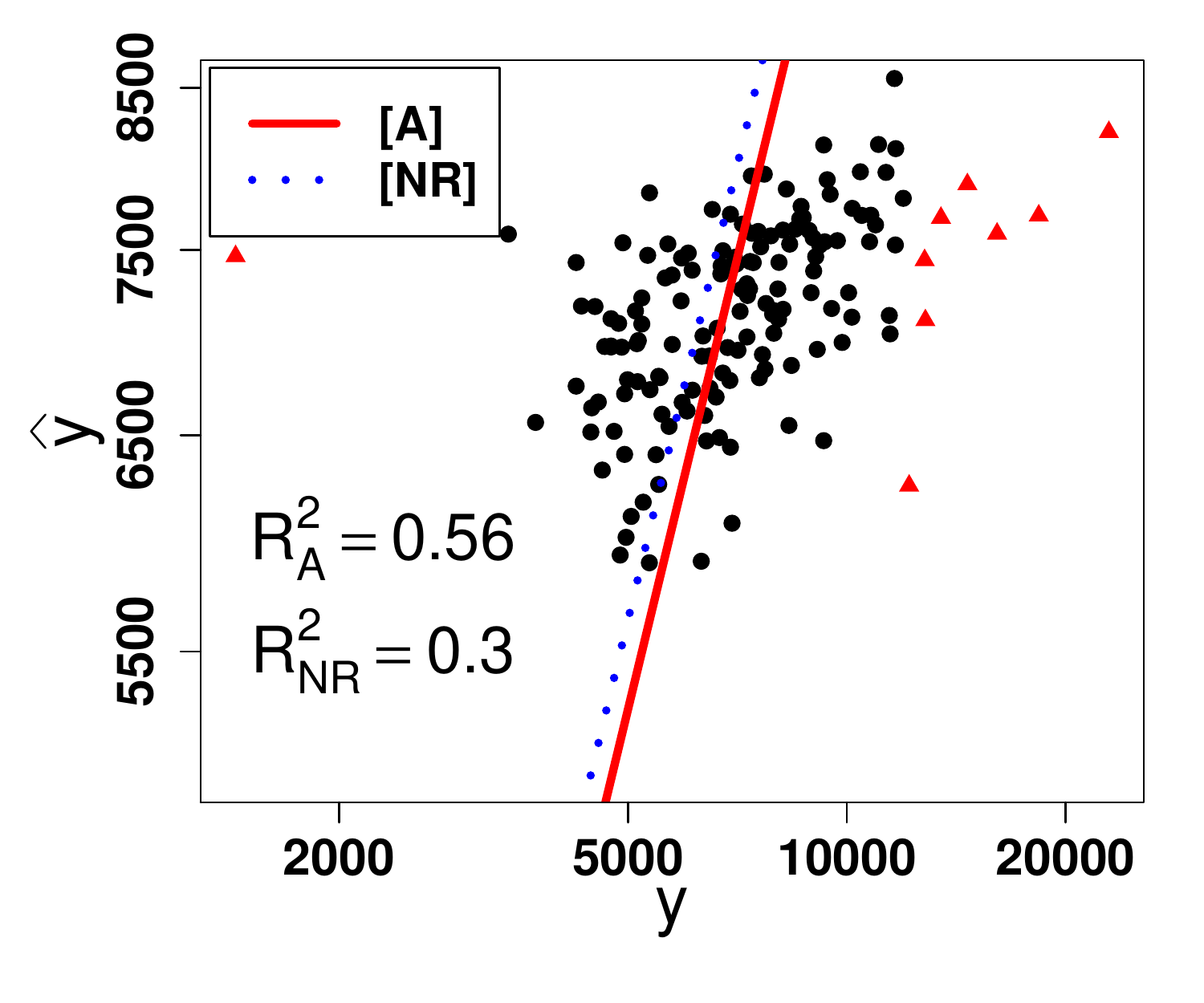}}
\hfill
\subfloat[E]{%
\includegraphics[page = 2, width=0.32\textwidth]{hiv_selbal_fit1.pdf}}
\hfill
\subfloat[H]{%
\includegraphics[page = 3, width=0.32\textwidth]{hiv_selbal_fit1.pdf}}
     \caption{Application - Robust HIV data analysis for the subcomposition $\C$ due to phylum in the taxonomy: Comparison of the robust and  non-robust approach in terms of the model fit statistics $R^2$ and outliers identified (red). In the robust procedure using RobRegCC, $R^2 = 1 - \|\bepsilon_{\lambda_{cv}}\|^2/ \|\y - \bar{y}\|^2$ where   $\bepsilon_{\lambda_{cv}} = \y - \X\what{\btheta}_{\lambda_{cv}} - \what{\bgamma}_{\lambda_{cv}}$ and $\bar{y}$ mean of $\y$.}
\label{fig:hiv_fit__r3}
\end{figure}

\subsection{Details about the robust initialization}
\label{suppl:robust} 
Here, we discuss the  principal sensitive component (PSC) based analysis for the sparse log-contrast model (S-LCM) \citep{shi2016regression}. Computing the sensitivity $\R$ \eqref{eq:psc} for the least square estimator is trivial  \citep{pena1999fast} as it avoids the separate model fitting to obtain $\hat{y}_i$'s and   $\hat{y}_{i(j)}$'s. Interestingly, the log-contrast model (LCM) follows the linear model. But the same is not true for  the S-LCM, hence, computing $\R$  is nontrivial. To overcome the challenge, we identify the support of the S-LCM coefficient estimate, and  compute $\R$ for the subsequent LCM (non-sparse). Please refer Algorithm \ref{alg:initrobregcc} for the formula of calculating $\R$. 



For the analysis, consider setting the parameter $\tau \in (0,0.5)$ and obtain $m=n\tau$.  Suppose $[\u_1, \ldots,\u_{q}]$ denote the principal components of $\R$. \citet{pena1999fast} characterized the extreme observations in terms of the  value of entries in a  principal component of $\R$. 
Following PENSE, for each $\u_i$,  we generate three candidate subsamples by removing $m$ observations corresponding to the: I) largest $\u_i$; II) smallest $\u_i$; III) largest $|\u_i|$. Including the one with all observations, the protocol results in total $3*q+1$ candidate samples.  
For each candidate sample, we estimate the  coefficient of S-LCM using the default procedure specified in   \citet{shi2016regression}. 
Now, using the coefficient estimate, we evaluate the candidate samples in terms of the M-estimator of scale  \citep{rousseeuw2011robust}  of the residuals obtained on full sample.  Suppose the chosen candidate sample attains the minimum scale value $s_1$.  A potentially "clean subsample"  is then obtained after discarding observations with the residuals magnitude (on the full data)  greater than some $C's_1$. See  \citet{CohenFreue2017} for the choice of $C'$.






The analysis may not detect the low-leveraged outliers. To solve,  \cite{pena1999fast} suggested to iterate the process several times or until convergence. Final S-LCM coefficient estimate  on the clean subsample,  and the residuals on full data, are used as initial estimator of the RobRegCC model, i.e., $\ddot{\bdelta} = [\ddot{\bbeta}\trans \, \ddot{\bgamma}\trans]\trans$. We have summarized the initialization procedure in the Algorithm \ref{alg:initrobregcc}.

\begin{algorithm}[htp]
\small
\caption{Initialization via PSC analysis}
\begin{algorithmic}\label{alg:initrobregcc}
\STATE Given $\y,\Z$, $\h = [h_1, \ldots, h_n] =  \M{diag}(\bP_{\Z})$, $\tau \in (0,0.5)$, $C_1=2$. 
\STATE Denote sorted $\h$ by $[h_{(1)}, \ldots, h_{(n)}]$.
\STATE Choose method: $\bM = \{\M{S-LCM,\,\, LCM}\}$
\STATE Define index set $\mathcal{A}^{[0]} = \{ i \,\,: h_i \leq h_{(n\alpha_1)} \}$, scale $s_0 = 1e^8$ with  $\alpha \in (0.5,1)$. 
\REPEAT
\STATE $\{\mathcal{A}^{[i+1]}, s_{i+1}\} \gets  \M{PSC-Analysis}(\y, \Z, \mathcal{A}^{[i]}, \tau, C_1)$
\UNTIL convergence, $|s_{i+1} - s_{i}| \leq 10^{-4}$.
\RETURN Index set $\mathcal{A}$.
    \vspace{0.2cm}
\STATE    \textbf{Solution:} \\ $\ddot{\bbeta} \gets \bM( \y_{\mathcal{A}}, \Z_{\mathcal{A}} )$; For more details of S-LCM see  \citet{shi2016regression}. \\
  $\ddot{\bgamma} = \y -\Z\ddot{\bbeta}$
\RETURN $\ddot{\bdelta} = [\ddot{\bbeta}\trans \, \ddot{\bgamma}\trans]\trans$
  
\vspace{0.5cm}
\STATE \underline{$\M{PSC-Analysis}(\y, \Z, \mathcal{A}, \tau, C_1)$ } 
    \vspace{0.1cm}
\STATE $\bar{\bbeta} \gets \bM( \y_{\mathcal{A}}, \Z_{\mathcal{A}} )$, \quad $\bar{\bepsilon} \gets \y_{\mathcal{A}} - \Z_{\mathcal{A}}\bar{\bbeta} $, \quad  $\mathcal{B} = \{i: \bar{\beta}_i \neq 0\}$, \quad $n_1 = |\mathcal{A}|$, $m = n_1 \tau$.
    \vspace{0.2cm}
\STATE Denote subset matrix ${\Z_{\mathcal{A} \mathcal{B}}}$, projection matrix  $\H = \bP_{\Z_{\mathcal{A} \mathcal{B}}}$
\STATE Define $\W$ as $\M{diag}(\W) = \bar{\bepsilon}/(1 - \M{diag}(\H))$. \STATE Compute $\R = \H \W^2 \H$.
\STATE $\U = [\u_1, \ldots,\u_{q}] \gets \M{Principal components of } \R$.
\STATE \textsc{Solution :} $\ddot{\bbeta} = \bar{\bbeta}$, $\ddot{\bgamma} = \bar{\bepsilon}$ and scale $\ddot{s} = \M{M-estimator}(\bar{\bepsilon})$. 
 \FOR{$i \in \{1,\ldots,q\}$}
 \STATE $\{\mathcal{A}_1, \mathcal{A}_2, \mathcal{A}_3\} \gets $ index set discarding $m$ \textsf{ largest} $\u_i$, \textsf{ smallest} $\u_i$, and \textsf{ largest} $|\u_i|$.
  \FOR{$j \in \{1,2,3\}$}
           \STATE  $\bar{\bbeta} \gets \bM( \y_{\mathcal{A}_j}, \Z_{\mathcal{A}_j} )$, $\bar{\bepsilon} \gets \y_{\mathcal{A}_j} - \Z_{\mathcal{A}_j}\bar{\bbeta} $
           \IF{( $\M{M-estimator}(\bar{\bepsilon}) < \ddot{s}$ )}  
           \STATE \textsc{Solution :} $\ddot{\bbeta} = \bar{\bbeta}$, $\ddot{\bgamma} = \bar{\bepsilon}$ and scale $\ddot{s} = \M{M-estimator}(\bar{\bepsilon})$. 
           \ENDIF
    \ENDFOR    
    \ENDFOR
  \RETURN  $\{|\y- \Z\ddot{\bbeta}| < C_1\ddot{s}, \,\, \ddot{s}\}$.
\end{algorithmic}
\end{algorithm}

\subsection{Non-asymptotic analysis proofs}
\begin{proof}[Proof of Theorem \ref{sec:non-asy:th1}] 
We consider the following optimization problem for the non-asymptotic analysis:
\begin{align}
\{\what{\bgamma}_{\lambda_1},\what{\bbeta}_{\lambda_2}\} \equiv \argmin_{\{\bgamma,\bbeta\}} \,\,   f_{\lambda_1,\lambda_2}(\bbeta, \bgamma; \Z, \y)    \quad \M{s.t.} \quad \C\trans \bbeta = \0 \, , \notag 
\end{align}
where $f_{\lambda_1,\lambda_2}(\bbeta, \bgamma; \Z, \y) = \frac{1}{2} \| \y
- \Z \bbeta - \bgamma \|_2^2  + P_{\lambda_1}^1(\bgamma) + P_{\lambda_2}^2(\bbeta)$. For the ease of presentation, we drop the  subscript $\{\lambda_1,\lambda_2\}$ from $f_{\lambda_1,\lambda_2}(\cdot)$ and $\{\what{\bgamma}_{\lambda_1},\what{\bbeta}_{\lambda_2}\}$. Then, for the optimal solution $\{\what{\bbeta}, \what{\bgamma} \}$, we have 
$$
f(\what{\bbeta}_{\lambda}, \what{\bgamma}_{\lambda}; \Z, \y) \leq f(\bbeta, \bgamma; \Z, \y),
$$
where $\{\bbeta, \bgamma\} \in \mathbb{R}^{p + n}$ such that  $\C\trans \bbeta = \0$. On simplification, we write  the basic inequality in terms of the true model parameters $\{\bbeta^*,\bgamma^*\}$, specified in equation \eqref{eq:trumdl}, as 
\begin{align}
M(\what{\bbeta}_{\lambda} - \bbeta^*, \what{\bgamma}_{\lambda} - \bgamma^*) & \leq  M(\bbeta - \bbeta^*, \bgamma - \bgamma^*)  + 2 \langle \bepsilon , \Z \bs\Delta^{(\beta)} +\bs \Delta^{(\gamma)} \rangle \notag \\ 
& + P_{\lambda_1}^1(\bgamma) - P_{\lambda_1}^1(\what{\bgamma}) + P_{\lambda_2}^2(\bbeta) - P_{\lambda_2}^2(\what{\bbeta} ), \label{eq:basicinequality}
\end{align}
where $\bs{\Delta}^{(\beta)}  = \what{\bbeta} - \bbeta$ and $\bs\Delta^{(\gamma)} = \what{\bgamma} - \bgamma$. To  simplify further,  we use following lemma to bound the  stochastic term $\langle \bepsilon , \Z \bs\Delta^{(\beta)} + \bs\Delta^{(\gamma)} \rangle$.
\begin{lemma}\label{lemma:lm1}
Define hard threshold penalty $P_{\lambda}^{h}(u ) = (-u^2/2 + \lambda|u|)1_{|u|<\lambda} +  (u^2/2 )1_{|u|\geq\lambda}$.  Consider $\U \in \bbR^{n \times p}$,  $\C \in \bbR^{p \times k}$ and 
define $\Gamma_{T,S} = \big\{ (\bgamma , \bbeta) \in \bbR^n \otimes \bbR^p ; \, \T =  \bs\mcJ(\bgamma), \,\, \S =  \bs\mcJ(\bbeta), \,\, \C\trans\bbeta = \0 \big\}$ where operator $\bs\mcJ(\cdot)$ denote support index set with $s = |\S|$ such that $ 1<s <p$ and $t = |\T|$ such that $ 1<t <n$. Suppose tuning parameter $\lambda_1 = A\lambda_a$  and $\lambda_2 = A\lambda_b$ with $\lambda_a = \sigma\sqrt{\log{en}}$, $\lambda_b = \sigma\sqrt{\log{ep}}$, and $A = \sqrt{ab}A_1$  for a sufficiently large  $A_1$ satisfying $a \geq 2b > 0$. 
There exist constant $L$, $C'$, $c$ and parameter $v$, 
we have 
\begin{align*}
\sup_{(\bbeta,\bgamma) \in \Gamma_{T,S}} \Bigg\{  2 \langle \bepsilon , \U\bbeta + \bgamma \rangle -  \frac{1}{a}\| \U\bbeta + \bgamma \|_2^2 -  \frac{P_{\lambda_1}^{h}(\bgamma )}{b} &  -  \frac{P_{\lambda_2}^{h}(\bbeta)}{b}  - \\&   2a L \sigma^2 (2-k) 
\Bigg\} \geq a \sigma^2 v
\end{align*}
with probability at most $C' \exp(-c v)$. 
\end{lemma}
Using the  lemma, we bound the stochastic component  of the \textsf{basic inequality} \eqref{eq:basicinequality} as 
\begin{align*}
2 \langle \bepsilon , \Z \bs\Delta^{(\beta)} + \bs\Delta^{(\gamma)} \rangle  \leq \frac{1}{a} \|\Z\bs\Delta^{(\beta)} +  \bs\Delta^{(\gamma)}\|_2^2 +  \frac{1}{b} P_{\lambda_1}^{h}(\bs\Delta^{(\gamma)} )  &+  \frac{1}{b} P_{\lambda_2}^{h}(\bs\Delta^{(\beta)}) + \\ & R + 2a L \sigma^2 (2-k), 
\end{align*}
where $\frac{1}{a} \|\Z\bs\Delta^{(\beta)} + \bs\Delta^{(\gamma)}\|_2^2 \leq \frac{2M(\bbeta-\bbeta^*,  \bgamma - \bgamma^*)}{a}  + \frac{2 M(\what{\bbeta}-\bbeta^*,  \what{\bgamma} - \bgamma^*) }{a}$ and 
\begin{align*}
R =  \sup_{ \substack{1 \leq s \leq p \\ 1 \leq t \leq n} } \,\,\,  \sup_{(\bbeta,\bgamma) \in \Gamma_{T,S}}  \Bigg\{ 2 \langle \bepsilon , \Z \bs\Delta^{(\beta)} + \bs\Delta^{(\gamma)} \rangle - & \frac{1}{a} \| \Z \bs\Delta^{(\beta)} + \bs\Delta^{(\gamma)}  \|_2^2 - \\ &
\frac{1}{b} P_{\lambda_1}^{h}(\bs\Delta^{(\gamma)} )  -  \frac{1}{b} P_{\lambda_2}^{h}(\bs\Delta^{(\beta)}) \Bigg\}
\end{align*}
with expectation $\mathbb{E}(R) \leq a c \sigma^2 $. Using the result, we write the \textsf{basic inequality}  \eqref{eq:basicinequality} as
\begin{align}
\Big( 1- \frac{1}{a} \Big) &M(\what{\bbeta}  -\bbeta^*,   \what{\bgamma} - \bgamma^*) \leq  \Big( 1 + \frac{1}{a} \Big) M(\bbeta-\bbeta^*,  \bgamma - \bgamma^*)   + \label{eq:lemma1:main1}  \\
 &  \frac{1}{b} P_{\lambda_1}^{h}(\bs\Delta^{(\gamma)} )  +  \frac{1}{b} P_{\lambda_2}^{h}(\bs\Delta^{(\beta)}) +  2a L \sigma^2 (2-k) + R +   \notag \\ 
 & P_{\lambda_1}^1(\bgamma) - P_{\lambda_1}^1(\what{\bgamma}) + P_{\lambda_2}^2(\bbeta) - P_{\lambda_2}^2(\what{\bbeta} ). \notag
\end{align}
To simplify further, we use the inequality $ P_{\lambda_1}^{h}(\bs\Delta^{(\gamma)} ) \leq P_{\lambda_1}^{h}(\bgamma ) + P_{\lambda_1}^{h}(\what{\bgamma} )$ and $P_{\lambda_1}^{h}(\bgamma)  \leq   P_{\lambda_1}^1(\bgamma)$, and write 
\begin{align*}
P_{\lambda_1}^1(\bgamma) - P_{\lambda_1}^1(\what{\bgamma}) + \frac{1}{b} P_{\lambda_1}^{h}(\bs\Delta^{(\gamma)} )
&\leq \Big( 1+ \frac{1}{b} \Big)  P_{\lambda_1}^1(\bgamma)  + \Big(  \frac{1}{b} -1 \Big)  P_{\lambda_1}^1(\what{\bgamma}) \\
&\leq \Big( 1+ \frac{1}{b} \Big) P_{\lambda_1}^1(\bgamma). 
\end{align*}
Similarly, we simplify $ P_{\lambda_2}^2(\bbeta) - P_{\lambda_2}^2(\what{\bbeta} ) +\frac{1}{b} P_{\lambda_2}^{h}(\bs\Delta^{(\beta)})\leq \Big( 1+ \frac{1}{b} \Big) P_{\lambda_2}^2(\bbeta)$. Now, we use the result from \eqref{eq:lemma1:main1} and write the expression for the  prediction error  bound as 
\begin{align*}
\Big( 1- \frac{1}{a} \Big) M(\what{\bbeta}-\bbeta^*, & \what{\bgamma} - \bgamma^*) \leq  \Big( 1 + \frac{1}{a} \Big) M(\bbeta-\bbeta^*,  \bgamma - \bgamma^*)  + R + \\ 
&\Big( 1+ \frac{1}{b} \Big) \{ P_{\lambda_1}^1(\bgamma) + 
P_{\lambda_2}^2(\bbeta) \} + 2a L \sigma^2 (2-k) .
\end{align*}
Thus, the oracle bound on the prediction error is   
\begin{align*}
M(\what{\bbeta}-\bbeta^*,  \what{\bgamma} - \bgamma^*) \, \lesssim \, M(\bbeta-\bbeta^*,  \bgamma - \bgamma^*)   + P_{\lambda_1}^1(\bgamma) + P_{\lambda_2}^2(\bbeta)  + \sigma^2 (3-k),
\end{align*}
where $\lesssim$ means the inequality holds upto a multiplicative constant. 

\end{proof}

\begin{proof}[Proof of Theorem \ref{sec:non-asy:th2}:]
We follow the  proof of Theorem \ref{sec:non-asy:th1} to prove the result. The result corresponds to the \textsf{case II} with LASSO penalty, i.e., $P_{\lambda_1}^1(\bgamma) = \lambda_1|\bgamma|_1$  and $P_{\lambda_2}^2(\bbeta) = \lambda_2|\bbeta|_1$. Note that $P_{\lambda_1}^{h}(\bs\Delta^{(\gamma)} ) \leq \lambda_1 |\bs\Delta^{(\gamma)}|_1 $ and  $P_{\lambda_2}^{h}(\bs\Delta^{(\beta)}) \leq \lambda_2 |\bs\Delta^{(\beta)}|_1 $.
Now, consider $ \theta = 1/b$ and  simplify
\begin{align*}
P_{\lambda_1}^1(\bgamma) -  P_{\lambda_1}^1(\what{\bgamma}) + \frac{P_{\lambda_1}^{h}(\bs\Delta^{(\gamma)} )}{b} 
&\leq   P_{\lambda_1}^1(\bgamma) -  P_{\lambda_1}^1(\what{\bgamma}) +  \frac{\lambda_1 |\bs\Delta^{(\gamma)}|_1}{b}
\\
&=   \lambda_1  \Big\{ |\bgamma|_1  -  |\what{\bgamma}|_1 + \theta |\bs\Delta^{(\gamma)}|_1  \Big\} \
\\
&\leq   \lambda_1 \Big\{  |\bs\Delta_{\T}^{(\gamma)}|_1   -   |\bs\Delta_{\T^c}^{(\gamma)}|_1 + \theta |\bs\Delta_{\T}^{(\gamma)}|_1   +   \theta|\bs\Delta_{\T^c}^{(\gamma)}|_1  \Big\} 
\\
&\leq   \lambda_1 \Big\{ (1+\theta) |\bs\Delta_{\T}^{(\gamma)}|_1   -  (1-\theta) |\bs\Delta_{\T^c}^{(\gamma)}|_1   \Big\} 
\\
&=  \lambda_1 (1-\theta) \Big\{ (1+\nu) |\bs\Delta_{\T}^{(\gamma)}|_1   -   |\bs\Delta_{\T^c}^{(\gamma)}|_1   \Big\} 
\end{align*}
where $\theta = \nu/ (1 + \nu) $. Similarly, 
\begin{align*}
P_{\lambda_2}^2(\bbeta) -  P_{\lambda_2}^2(\what{\bbeta}) + \frac{P_{\lambda_2}^{h}(\bs\Delta^{(\beta)} )}{b} 
\leq \lambda_2 (1-\theta) \Big\{ (1+\nu) |\bs\Delta_{\S}^{(\beta)}|_1   -   |\bs\Delta_{\S^c}^{(\beta)}|_1   \Big\}.
\end{align*}
On combining the two results, we get 
\begin{align*}
&P_{\lambda_2}^2(\bbeta) -  P_{\lambda_2}^2(\what{\bbeta}) + \frac{P_{\lambda_2}^{h}(\bs\Delta^{(\beta)} )}{b}   +  P_{\lambda_1}^1(\bgamma) -  P_{\lambda_1}^1(\what{\bgamma}) + \frac{P_{\lambda_1}^{h}(\bs\Delta^{(\gamma)} )}{b} \leq \\&  \lambda_2 (1-\theta) \Big\{ (1+\nu) |\bs\Delta_{\S}^{(\beta)}|_1   -   |\bs\Delta_{\S^c}^{(\beta)}|_1   \Big\} + \lambda_1 (1-\theta) \Big\{ (1+\nu) |\bs\Delta_{\T}^{(\gamma)}|_1   -   |\bs\Delta_{\T^c}^{(\gamma)}|_1   \Big\} .
\end{align*}
Under the compatibility condition on RHS, we further simplify
\begin{align*}
\frac{\M{RHS}}{1-\theta} &\leq \lambda_1  \kappa_1 t^{1/2} \| \bP_{\Z\bs\Delta^{(\beta)}}^{\perp} (\Z\bs\Delta^{(\beta)} + \bs\Delta^{(\gamma)}) \|_2 + \lambda_2  \kappa_2 s^{1/2} \| \bP_{\Z\bs\Delta^{(\beta)}} (\Z\bs\Delta^{(\beta)} + \bs\Delta^{(\gamma)}) \|_2 \\ \M{RHS}
&\leq \lambda_1 (1-\theta) \kappa_1 t^{1/2} \| \Z\bs\Delta^{(\beta)} + \bs\Delta^{(\gamma)} \|_2 + \lambda_2 (1-\theta) \kappa_2 s^{1/2} \| \Z\bs\Delta^{(\beta)} + \bs\Delta^{(\gamma)} \|_2 \\
&\leq \frac{2}{a}  \| \Z\bs\Delta^{(\beta)} + \bs\Delta^{(\gamma)} \|_2^2  +  a\lambda_1^2 (1-\theta)^2 \kappa_1^2 t   +  a\lambda_2^2 (1-\theta)^2 \kappa_2^2 s
\end{align*}
Using the upper bound $\| \Z\bs\Delta^{(\beta)} + \bs\Delta^{(\gamma)} \|_2^2 \leq M(\what{\bbeta}-\bbeta^*,  \what{\bgamma} - \bgamma^*) + M(\bbeta-\bbeta^*,  \bgamma - \bgamma^*) $, and 
following the proof of Theorem \ref{sec:non-asy:th1}, we  write 
\begin{align*}
\Big( 1- \frac{2}{a} \Big) M(\what{\bbeta}-\bbeta^*,  \what{\bgamma} - \bgamma^*) \leq & \Big( 1 + \frac{2}{a} \Big) M(\bbeta-\bbeta^*,  \bgamma - \bgamma^*)  + 2aL\sigma^2(2-k) 
\\ & + a\lambda_1^2 (1-\theta)^2 \kappa_1^2 t   +  a\lambda_2^2 (1-\theta)^2 \kappa_2^2 s +R.
\end{align*}
Thus,  we can say that
$$
 M(\what{\bbeta}-\bbeta^*,  \what{\bgamma} - \bgamma^*) \lesssim    M(\bbeta-\bbeta^*,  \bgamma - \bgamma^*)  + a(1-\theta)^2 \{ \lambda_1^2  \kappa_1^2 t   +  \lambda_2^2  \kappa_2^2 s \} + (3-k)\sigma^2.
$$
\end{proof}

\begin{proof}[Proof of lemma \ref{lemma:lm1}] We follows the approach  from \citet{Shea,She2017b,she2017selective} to prove the result. We write the forms of hard threshold penalty as $P_{ \lambda}^{h}(u) = (-u^2/2 + \lambda|u|)1_{|u|<\lambda} +  (u^2/2 )1_{|u|\geq\lambda}$ and $P_{ \lambda}^{0}(u ) = (u^2/2])1_{u \neq 0}$. 
Define 
\begin{align*}
L_h(\bbeta,\bgamma) &= 2 \langle \bepsilon , \U\bbeta + \bgamma \rangle - \frac{\| \U\bbeta + \bgamma \|_2^2}{a} - \frac{ P_{\lambda_1}^{h}(\bgamma )}{b} - \frac{ P_{\lambda_2}^{h}(\bbeta) }{b} -  2a L \sigma^2 (2-k)\\
L_0(\bbeta,\bgamma) &= 2 \langle \bepsilon , \U\bbeta + \bgamma \rangle - \frac{\| \U\bbeta + \bgamma \|_2^2}{a} - \frac{ P_{\lambda_1}^{0}(\bgamma )}{b} - \frac{ P_{\lambda_2}^{0}(\bbeta) }{b} -  2a L \sigma^2 (2-k).
\end{align*}
The formulation implies $P_{\lambda}^{h}(\bgamma) \leq P_{\lambda}^{0}(\bgamma)$ resulting in  $L_h(\bbeta,\bgamma) \geq L_0(\bbeta,\bgamma)$. 
Define set 
\begin{align*}
\mcA_h = \Bigg \{  \sup_{(\bbeta,\bgamma) \in \Gamma_{T,S}} L_h(\bbeta,\bgamma) \geq a\sigma^2 v  \Bigg \},\,\,
\mcA_0 = \Bigg \{  \sup_{(\bbeta,\bgamma) \in \Gamma_{T,S}} L_0(\bbeta,\bgamma) \geq a\sigma^2 v  \Bigg \}. 
\end{align*}
For any $\zeta \in \mcA_0$, the formulation implies $\zeta \in \mcA_H$, hence, $\mcA_0 \subseteq \mcA_H$. To prove $\mcA_0 = \mcA_H$, we aim to prove $\mcA_H \subseteq \mcA_0$. It should be noted that 
$$
\mcA_h \subset  \Big \{  \sup_{(\bbeta,\bgamma)} L_h(\bbeta,\bgamma) \geq a\sigma^2 v  \Big \} . 
$$
The occurrence of $\mcA_h$ implies that 
$L_h(\bbeta^0,\bgamma^0) \geq a\sigma^2 v$
for any $(\bbeta^0,\bgamma^0)$ satisfying 
$$
(\bbeta^0,\bgamma^0) \equiv \argmin_{(\bbeta,\bgamma)} \,\, \frac{ \| \U\bbeta + \bgamma \|_2^2}{a}- 2 \langle \bepsilon, \U\bbeta+\bgamma\rangle + \frac{ P_{\lambda_1}^{h}(\bgamma )}{b} + \frac{ P_{\lambda_2}^{h}(\bbeta) }{b} 
$$

\begin{lemma}[For proof see \citet{she2012iterative}]\label{lemma:lm2}
Suppose $\alpha > 1$. For $\y \in \mathbb{R}^n$, there exist a globally optimal solution $\bgamma^0$ satisfying 
$$
\bgamma^0 \equiv \argmin_{\bgamma} \,\, \frac{1}{2} \| \y - \bgamma\|_2^2 + \alpha  P_{\lambda}^h(\bgamma ),
$$
such that for any $j: 1 \leq j \leq n$ either $\gamma_j^0 = 0$ or $|\gamma_j^0| \geq \lambda \alpha^{1/2} \geq \lambda$. 
\end{lemma}
Lemma \ref{lemma:lm2} and Lemma 5 from \citet{Shea} indicate that a globally optimal solution $(\bbeta^0,\bgamma^0)$ exist. Also, in a case with $a > 2b>0$, we have 
$P_{0}(\bgamma ; \lambda_1) +  P_{0}(\bbeta ; \lambda_2) =  P_{h}(\bgamma ; \lambda_1) +  P_{h}(\bbeta ; \lambda_2)$,  and thus $ L_0(\bbeta^0,\bgamma^0)   = L_h(\bbeta^0,\bgamma^0)$. Moreover, 
$$
\sup_{(\bbeta,\bgamma) \in \Gamma_{T,S}} L_0(\bbeta,\bgamma) \geq L_0(\bbeta^0,\bgamma^0) = L_h(\bbeta^0,\bgamma^0) \geq a \sigma^2 v ,
$$
suggests $\mcA_h \subseteq \mcA_0$. Hence, $\mcA_h = \mcA_0$. It is then sufficient to prove that 
$$
\bP(\mcA_h) = \bP(\mcA_0) \leq C' \exp(-cv).
$$

Now, let $\I_n$ be the identity matrix of size $n \times n$, and $\I_{\S}$ be the sub-matrix  corresponding to columns in the index set $\S$. $\bP_{\I_{\S}}$ denote the projection matrix for  the  sub-matrix $\I_{\S}$.  Then, $\bP_{\I_{S}} +\bP_{\I_{S}}\trans= \I_n$. In terms of the projection matrix, we factorize the stochastic  component 
\begin{align*}
\langle \bepsilon, \U\bbeta + \bgamma \rangle =  \langle \bepsilon, \bP_{\I_{S}}\trans \U\bbeta \rangle + \langle \bepsilon, \bP_{\I_{S}} (\U\bbeta + \bgamma) \rangle  = \langle \bepsilon, \a_1\rangle + \langle \bepsilon, \a_2 \rangle , 
\end{align*}
where $\|\a_1\|_2^2 + \|\a_2\|_2^2 =\| \U\bbeta + \bgamma \|_2^2$. 

\begin{lemma}\label{lemma:3}
Given $\U \in \bbR^{n \times p}$, $\C \in \bbR^{p \times k}$, index set $\S$ with $s = |\S|$ such that $ 1<s <p$, and index set $\T$ with $t = |\T|$ such that $ 1<t <n$. 
Define set $\Gamma_{T,S}^' = \big\{ \bs\alpha \in \bbR^n ; \,\, \|\bs\alpha\|_2 \leq 1 , \bs\alpha = \U\bbeta, \bbeta \in \mathbb{R}^p,  \bs\alpha \in \M{CS}(\U_{\T^c \S}), \, \T =  \bs\mcJ(\bgamma), \,\, \S =  \bs\mcJ(\bbeta), \,\, \C\trans\bbeta = \0 \big\}$ where operator $\bs\mcJ(\cdot)$ denote the support set. Define  $p_1^'(t,s) = \sigma^2\{s-k  + \log{n \choose t} + \log{p \choose k}  + \log{p-k \choose s-k} \}$. Then
$$
\bP\Bigg( \sup_{\bs\alpha \in \Gamma_{T,S}^'}  \langle \bepsilon , \bs\alpha \rangle \geq v\sigma + \sqrt{L p_1^'(t,s) } \Bigg) \leq C' \exp (-c v^2),
$$
for sufficiently large constant $\{L,C', c \}$.




\end{lemma}
Proof of Lemma \ref{lemma:3} follows from the Lemma 6 of \citet{Shea} and the Lemma 4  of \citet{she2017selective}. Now, using the lemma \ref{lemma:3}, we simplify the first term involving $\a_1$. Thus, write 
\begin{align*}
2 \langle \bepsilon , \a_1 \rangle - \frac{1}{a} \|\a_1\|^2 &-2a L p_1^'(t,s) 
= 2 \langle \bepsilon , {\a_1}  \rangle - \frac{ \|\a_1\|^2}{2a}   - \frac{ \|\a_1\|^2}{2a}  - 2a L p_1^'(t,s) \\
&\leq 2 \langle \bepsilon , \frac{\a_1}{\|\a_1\|}  \rangle \|\a_1\|  - \frac{ \|\a_1\|^2}{2a}  - 2 \|\a_1\|  \{L p_1^'(t,s) \}^{1/2},
\end{align*}
on applying the Cauchy-Schwarz inequality on the last two terms. Further, simplify RHS as 
\begin{align*}
\M{RHS} &= 2 \|\a_1\| \Big ( \langle \bepsilon , \frac{\a_1}{\|\a_1\|}  \rangle -  (L p_1^'(t,s) )^{1/2}   \Big ) - \frac{1}{2a} \|\a_1\|^2 \\
&\leq 2a \Big ( \langle \bepsilon , \frac{\a_1}{\|\a_1\|}  \rangle -  (L p_1^'(t,s) )^{1/2}   \Big )_+^2 +   \frac{1}{2a} \|\a_1\|^2  - \frac{1}{2a} \|\a_1\|^2 \\
&= 2a \Big ( \langle \bepsilon , \frac{\a_1}{\|\a_1\|_2}  \rangle -  (L p_1^'(t,s) )^{1/2}   \Big )_+^2
\end{align*}
Again, using the result from lemma \ref{lemma:3}, we have
\begin{align}
&\bP\Bigg \{ \sup_{(\bbeta,\bgamma) \in \Gamma_{T,S}}   \Big (   2  \langle \bepsilon , \a_1 \rangle - \frac{\|\a_1\|^2}{a}  -2a L p_1^'(t,s) \Bigg )  \geq \frac{1}{2} a \sigma^2 v \Big \} \leq  \notag \\ 
 &\bP\Bigg \{ \sup_{(\bbeta,\bgamma) \in \Gamma_{T,S}}   2a \Big ( \langle \bepsilon , \frac{\a_1}{\|\a_1\|}  \rangle -  \{L p_1^'(t,s) \}^{1/2}   \Big )_+^2  \geq \frac{1}{2} a \sigma^2 v \Bigg \} \leq C' \exp (-cv) . \label{eq:lemma1:1}
\end{align}
Similarly for $\a_2$, we have 
\begin{align}
&\bP\Big \{ \sup_{(\bbeta,\bgamma) \in \Gamma_{T,S}}   \Big (   2  \langle \bepsilon , \a_2 \rangle - \frac{1}{a} \|\a_2\|^2 -2a L p_2^'(t,s) \Big )  \geq \frac{1}{2} a \sigma^2 v \Big \} \leq C' \exp (-cv) ,  \label{eq:lemma1:2}
\end{align}
for $p_2^'(t,s) =  \sigma^2\{t + \log{n \choose t} \}$. On applying the union bound on the results obtained in \eqref{eq:lemma1:1} and \eqref{eq:lemma1:2}, we have 
\begin{align}
\bP\Bigg \{ \sup_{(\bbeta,\bgamma) \in \Gamma_{T,S}}   \Bigg [  2  \langle \bepsilon , \U\bbeta + \bgamma \rangle &- \frac{ \|\U\bbeta + \bgamma \|_2^2 }{a}-2a L \sigma^2 \Bigg \{ s-k+t+2\log{n \choose t} +\notag \\ & \log{p \choose k} +\log{p-k \choose s-k} \Bigg\} \Bigg ]   \geq  a \sigma^2 v \Bigg \} \notag \\ & \leq C' \exp (-cv), \label{eq:subnonasyunionbd}
\end{align}
for some constants $\{L, C', c \}$. Now, for the sufficiently  large constant $A_1$ and $a \geq 2b > 0$, we can say that 
$$
 2a L \sigma^2 \Big\{ t+ 2\log{n \choose t}  \Big\} \leq 2a L \sigma^2 \Big\{ t+ 2t \log{\frac{en}{t}}  \Big\} \leq 4a L \sigma^2 t\log(en)\leq \frac{ P_{\lambda_1}^h(\bgamma)}{b}.
$$
Similarly, 
\begin{align*}
 s-k+ \log{p \choose k}  &+ \log{p-k \choose s-k}  \leq  s-k+ k\log{\frac{ep}{k}}  + (s-k)\log{ \frac{ep}{s-k}} \\ & 
=   s\log{ep} +(s-k)  - (s-k)\log{(s-k)} -k \log{k}   \\ & 
 \leq  s\log{ep}  +(s-k)   -(s-k)\Big(1-\frac{1}{s-k} \Big) -k\Big(1-\frac{1}{k} \Big) 
 \\ & = s\log{ep} + 2-k ,
\end{align*}
where second inequality is due to $1-1/m \leq \log{m}$. Thus, for sufficiently large $A_1$, we write  
\begin{align*}
  2a L \sigma^2 \Big\{ s-k+ \log{p \choose k}  + \log{p-k \choose s-k} \Big\} & \leq   2a L \sigma^2 (s\log{ep} + 2-k ) \\ & \leq \frac{1}{b} P_h(\bbeta;\lambda_2) + 2a L \sigma^2 (2-k).
\end{align*}
On applying the above inequality results in the union bound \eqref{eq:subnonasyunionbd}, we prove that
\begin{align*}
\bP\Bigg ( \sup_{(\bbeta,\bgamma) \in \Gamma_{T,S}}  \,\, \Bigg \{ 2  \langle \bepsilon , \X\bbeta + \bgamma \rangle -& \frac{1}{a} \|\X\bbeta + \bgamma \|_2^2 -  \frac{P_{\lambda_1}^{h}(\bgamma )}{b}  - \frac{P_{\lambda_2}^{h}(\bbeta)}{b}  - \\ &  2a L \sigma^2 (2-k)
\Bigg \}   \geq a \sigma^2 v \Bigg )  \leq C' \exp (-cv) .
\end{align*}
\end{proof}

\end{document}